\documentclass[lettersize,journal]{IEEEtran}
\usepackage{amsmath,amsfonts,amssymb}

\usepackage{mathtools}
\usepackage{amsthm}
\newtheorem{assumption}{Assumption}
\newtheorem{lemma}{Lemma}
\newtheorem{theorem}{Theorem}
\newtheorem{remark}{Remark}
\newtheorem{corollary}{Corollary}
\newtheorem{definition}{Definition}
\usepackage{algorithm}
\usepackage{algpseudocode}
\usepackage[utf8]{inputenc}
\usepackage{array}
\usepackage[caption=false,font=footnotesize]{subfig}
\usepackage{textcomp}
\usepackage{stfloats}
\usepackage{url}
\usepackage{verbatim}
\usepackage{graphicx}
\usepackage{cite}
\newcommand{\R}{\mathbb{R}}

\newcommand{\T}{^\top}
\newcommand{\con}{\mathrm{con}}
\DeclareMathOperator*{\minimize}{minimize}
\DeclareMathOperator{\st}{subject~to}
\usepackage{tikz}
\usetikzlibrary{shapes.geometric, arrows.meta, positioning, fit, backgrounds}

\begin{document}

\title{Adaptive Trajectory Bundle Method for Roll-to-Roll Manufacturing Systems}

\author{\IEEEauthorblockN{Jiachen Li
, Shihao Li, 
Christopher Martin, Wei Li, Dongmei Chen}

\IEEEauthorblockA{The University of Texas at Austin\\
\{jiachenli, shihaoli01301, cbmartin129\}@utexas.edu},{weiwli@austin.utexas.edu},{dmchen@me.utexas.edu}
\thanks{This paper was produced by Advanced Power Systems and Control Laboratory. They are in UT Austin,TX.}
\thanks{Manuscript received Dec 29, 2025; revised Dec 29, 2025.}}

\markboth{Journal of \LaTeX\ Class Files,~Vol.~14, No.~8, August~2021}%
{Shell \MakeLowercase{\textit{et al.}}: A Sample Article Using IEEEtran.cls for IEEE Journals}


\maketitle

\begin{abstract}
Roll-to-roll (R2R) manufacturing requires precise tension and velocity control under operational constraints. Model predictive control demands gradient computation, while sampling-based methods like MPPI struggle with hard constraint satisfaction. This paper presents an adaptive trajectory bundle method that achieves rigorous constraint handling through derivative-free sequential convex programming. The approach approximates nonlinear dynamics and costs via interpolated sample bundles, replacing Taylor-series linearization with function-value interpolation. Adaptive trust region and penalty mechanisms automatically adjust based on constraint violation metrics, eliminating manual tuning. We establish convergence guarantees proving finite-time feasibility and convergence to stationary points of the constrained problem. Simulations on a six-zone R2R system demonstrate that the adaptive method achieves 4.3\% lower tension RMSE than gradient-based MPC and 11.1\% improvement over baseline TBM in velocity transients, with superior constraint satisfaction compared to MPPI variants. Experimental validation on an R2R dry transfer system confirms faster settling and reduced overshoot relative to LQR and non-adaptive TBM.
\end{abstract}


\section{Introduction}

Roll-to-roll (R2R) manufacturing processes are essential for high-throughput production of flexible electronics, photovoltaics, functional films, and advanced materials. These systems transport continuous web material through sequential processing zones, where precise tension and velocity regulation are critical to product quality. Excessive tension causes web breakage, while insufficient tension leads to wrinkling and defects. The control challenge is compounded by strong coupling between adjacent zones, time-varying parameters due to changing roll radii, and strict operational constraints on tensions, velocities, and actuator rates.

Traditional control approaches face fundamental limitations. Model Predictive Control (MPC), while capable of handling constraints, requires accurate gradient information and suffers from computational burden in high-dimensional settings. Sampling-based methods like Model Predictive Path Integral (MPPI) control avoid gradient computation but struggle with hard constraint satisfaction. This motivates the need for a framework combining constraint handling with derivative-free optimization.

\subsection{Related Work}

Classical R2R control relies on PID, LQR, and MPC variants to regulate web tension and transport speed. Foundational modeling established tension/velocity decoupling strategies \cite{koc2002winding}, while industrial studies refined adaptive PI/PID implementations for operating-point drift \cite{raul2015adaptivePI,chen2017fuzzy}. Recent work emphasizes advanced control methods \cite{martin2024stabilization,martin2025sequential} and physics-consistent tension models \cite{jeong2021tension,he2024multispan}. Manufacturing surveys underline the role of monitoring and feedback design in achieving system robustness \cite{martin2022hinfty,martin2025hinfty}. Data-driven approaches report robust constrained tension regulation \cite{chen2023robustR2R} and AI-assisted tension estimation \cite{gafurov2024webtension,gafurov2025aidt}, with reviews pointing to hybrid controllers as a promising direction \cite{martin2025review}.

Despite these advances, classical methods require extensive tuning and provide limited guarantees under coupled nonlinear dynamics. Advanced MPC formulations demand gradient computation and careful initialization, with computational requirements becoming prohibitive for systems with many zones or long prediction horizons.

MPPI control is inherently unconstrained, typically relying on soft-constraint penalties \cite{park2025csc, yin2023shield}. Methods to enforce hard constraints include CBF-based shields \cite{yin2023shield}, though these can be myopic \cite{rabiee2024guaranteed}. Safe-by-construction approaches such as GS-MPPI embed composite CBFs into dynamics \cite{rabiee2024guaranteed}, while MPPI-DBaS augments states with discrete barrier states \cite{wang2025mppi}. DualGuard-MPPI integrates Hamilton-Jacobi reachability for provable safety \cite{borquez2025dualguard}. Other approaches improve sample efficiency using learned surrogates (BC-MPPI) \cite{ezeji2025bc}, specialized samplers \cite{yan2023output,mpopi_legged_2025}, or clustering algorithms to address multimodality \cite{park2025csc}.

While these MPPI variants succeed in robotics, they face limitations for R2R manufacturing. The single-shooting formulation makes long-horizon planning unstable for coupled multi-zone dynamics, and existing constraint mechanisms either provide soft guarantees, require additional filtering layers, or demand specialized model structures.

Recently, the Trajectory Bundle Method (TBM) introduced a derivative-free optimal control paradigm based on sequential convex programming \cite{tracy2025trajectory}. Instead of Taylor-series approximations, TBM uses interpolated bundles of sampled trajectories to approximate cost, dynamics, and constraint functions.

\subsection{Contributions}

This paper develops an adaptive trajectory bundle method for constrained R2R control with three main contributions: (1) a complete TBM formulation for R2R tension-velocity dynamics with asymmetric penalty handling for over/under-tensioning (Algorithm~\ref{alg:adaptive_tbm}, Figure~\ref{fig:flowchart}); (2) adaptive trust region and penalty scheduling that automatically adjust based on constraint violations, eliminating manual tuning; and (3) convergence guarantees proving finite-time feasibility and convergence to stationary points.

The paper is organized as follows: Section~II derives R2R dynamics, Section~III presents the TBM framework, Section~IV provides simulation and experimental results, and Section~V concludes.

\begin{figure}[t]
\centering
\begin{tikzpicture}[
    node distance=0.35cm,
    startstop/.style={rectangle, rounded corners, minimum width=2.2cm, minimum height=0.45cm, text centered, draw=black, fill=green!20, font=\footnotesize},
    process/.style={rectangle, rounded corners, minimum width=2.2cm, minimum height=0.45cm, text centered, draw=black, fill=blue!15, font=\footnotesize},
    adapt/.style={rectangle, rounded corners, minimum width=2.2cm, minimum height=0.45cm, text centered, draw=black, fill=orange!20, font=\footnotesize},
    decision/.style={diamond, minimum width=1.5cm, minimum height=0.6cm, text centered, draw=black, fill=yellow!20, aspect=2, font=\footnotesize},
    arrow/.style={-Stealth}
]

\node (start) [startstop] {Initialize $z^0, \Delta^0, \mu^0, \boldsymbol{\gamma}^0$};
\node (sample) [process, below=of start] {Bundle Generation};
\node (build) [process, below=of sample] {Build Bundle Matrices};
\node (solve) [process, below=of build] {Solve Convex Subproblem};
\node (recover) [process, below=of solve] {Recover Trajectory};
\node (violation) [process, below=of recover] {Compute Violations};
\node (decide) [decision, below=of violation] {Converged?};
\node (stop) [startstop, above left=0.5cm and 0.8cm of decide] {Return $z^*$};
\node (adapt1) [adapt, above right=0.15cm and 0.5cm of decide] {Trust Region Adapt};
\node (adapt2) [adapt, above=of adapt1] {Penalty Adapt};

\draw [arrow] (start) -- (sample);
\draw [arrow] (sample) -- (build);
\draw [arrow] (build) -- (solve);
\draw [arrow] (solve) -- (recover);
\draw [arrow] (recover) -- (violation);
\draw [arrow] (violation) -- (decide);
\draw [arrow] (decide) -| node[anchor=south, pos=0.25, font=\footnotesize] {Yes} (stop);
\draw [arrow] (decide) -| node[anchor=south, pos=0.25, font=\footnotesize] {No} (adapt1);
\draw [arrow] (adapt1) -- (adapt2);

\draw [arrow] (adapt2.east) -- ++(0.4,0) |- (sample.east);

\end{tikzpicture}
\caption{Flowchart of the Adaptive Trajectory Bundle Method.}
\label{fig:flowchart}
\end{figure}
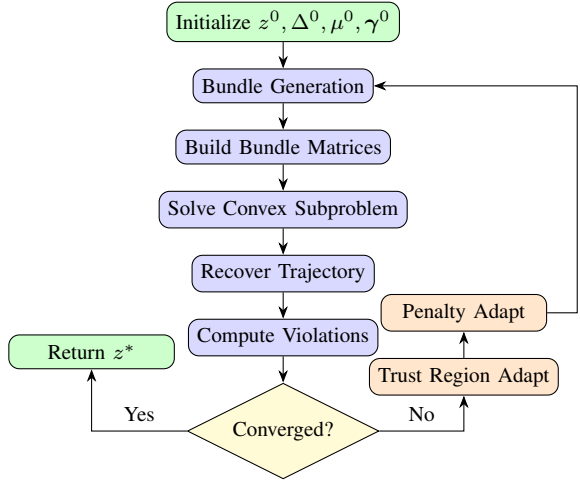

\section{Roll-to-Roll System Dynamics}

Consider a R2R manufacturing line with $N$ motorized rollers indexed by $i \in \{1, \ldots, N\}$, with continuous web material spanning between consecutive rollers, as shown in \ref{fig:R2R}. We assume: (1) passive rollers are omitted; (2) web tension $T_i$ is spatially uniform between adjacent rollers; (3) no slippage occurs, giving $v_i = \omega_i R_i$; and (4) the web exhibits linear elastic behavior with Young's modulus $E$ and cross-sectional area $A$.

\begin{figure}[htbp]
    \centering
    \includegraphics[width=0.35\textwidth]{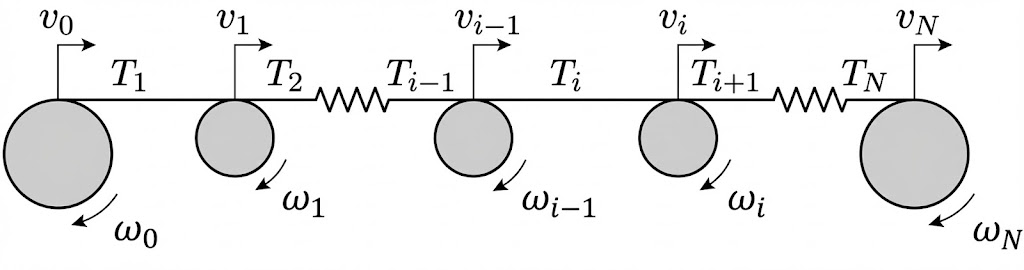}
    \caption{Schematic of a simplified R2R line.}
    \label{fig:R2R}
\end{figure}

The tension dynamics in section $i$ follow from mass conservation:
\begin{equation}
\frac{\mathrm{d}T_i}{\mathrm{d}t} = \frac{EA}{L_i}(v_i - v_{i-1}) + \frac{1}{L_i}(T_{i-1}v_{i-1} - T_i v_i),
\label{eq:tension_dynamics}
\end{equation}
where $L_i$ is the span length and boundary conditions $T_0 = T_{N+1} = 0$ apply. The roller velocity dynamics result from torque balance:
\begin{equation}
\frac{\mathrm{d}v_i}{\mathrm{d}t} = \frac{R_i^2}{J_i}(T_{i+1} - T_i) - \frac{f_i}{J_i}v_i + \frac{R_i}{J_i}u_i + b_i \frac{\mathrm{d}w_i}{\mathrm{d}t},
\label{eq:velocity_dynamics}
\end{equation}
where $J_i$ is roller inertia, $f_i$ is viscous friction, $u_i$ is motor torque, and $w_i$ is a Brownian motion capturing stochastic disturbances.

For trajectory optimization, we discretize \eqref{eq:tension_dynamics}--\eqref{eq:velocity_dynamics} using Euler-Maruyama integration with sampling period $\Delta t$. Define state $x_t = [T_1, \ldots, T_N, v_1, \ldots, v_N]\T \in \R^{2N}$ and control $u_t = [u_1, \ldots, u_N]\T \in \R^N$. The deterministic propagation map used for bundle construction is:
\begin{equation}
F(x_t, u_t) \triangleq x_t + \big( f(x_t) + G u_t \big) \Delta t,
\label{eq:deterministic_map}
\end{equation}
where $f: \R^{2N} \to \R^{2N}$ encodes the drift terms from \eqref{eq:tension_dynamics}--\eqref{eq:velocity_dynamics} and $G \in \R^{2N \times N}$ maps control inputs to velocity states.
\section{Adaptive Trajectory Bundle Method: Theoretical Framework}
\label{sec:tbm_formulation}

This section develops the adaptive trajectory bundle method for R2R manufacturing control. We first review derivative-free function approximation via trajectory bundles, then formulate the convex subproblem, and finally introduce adaptive mechanisms for trust region and penalty management that eliminate manual tuning.

\subsection{Notation and Preliminaries}

Let $N$ denote the number of motorized rollers and $H$ the prediction horizon. The state dimension is $n_x = 2N$ (tensions and velocities) and control dimension is $n_u = N$ (motor torques). A trajectory over horizon $H$ is denoted $z = \{(x_k, u_k)\}_{k=1}^{H}$ where $x_k \in \R^{n_x}$ and $u_k \in \R^{n_u}$. We endow the trajectory space with the norm:
\begin{equation}
\|z\| = \sqrt{\sum_{k=1}^{H} \left( \|x_k\|^2 + \|u_k\|^2 \right)}.
\label{eq:traj_norm}
\end{equation}

The probability simplex is denoted $\Delta^{m-1} = \{\alpha \in \R^m : \sum_{i=1}^m \alpha_i = 1, \alpha_i \geq 0\}$, and $B(z, \Delta) = \{z' : \|z' - z\| \leq \Delta\}$ is the closed ball of radius $\Delta$ centered at $z$. The operator $[\cdot]_- : \R^n \to \R^n_{\geq 0}$ is defined by $[v]_- = \max(0, -v)$ applied element-wise, extracting the negative part so that $[c]_- > 0$ indicates constraint violation when $c \geq 0$ is required.

\subsection{Derivative-Free Function Approximation via Bundles}

The trajectory bundle method approximates nonlinear functions through interpolation of sampled values rather than Taylor series expansion. This approach requires only function evaluations, making it suitable for systems where analytical gradients are expensive or unavailable \cite{conn2009derivative}. The bundle method terminology originates from nonsmooth convex optimization \cite{lemarechal1975,kiwiel1985}, where function values and subgradients are collected into ``bundles'' to construct piecewise-linear approximations. The trajectory bundle method \cite{tracy2025trajectory} adapts this philosophy to trajectory optimization using function value interpolation rather than subgradient cutting planes.

\subsubsection{Bundle Interpolation}

Consider an arbitrary function $p: \R^a \to \R^b$. Given a reference point $\bar{y} \in \R^a$ and a sampling radius $\Delta > 0$, we collect $m$ evaluation points $\{y_1, \ldots, y_m\} \subset B(\bar{y}, \Delta)$ and compute $p_i = p(y_i)$ for each sample. These are arranged into matrices:
\begin{equation}
W_y = \begin{bmatrix} y_1 & \cdots & y_m \end{bmatrix} \in \R^{a \times m}, \quad
W_p = \begin{bmatrix} p_1 & \cdots & p_m \end{bmatrix} \in \R^{b \times m}.
\label{eq:bundle_matrices}
\end{equation}

The bundle approximation represents candidate solutions as convex combinations of samples \cite{conn2009derivative}:
\begin{equation}
y = W_y \alpha, \quad p(y) \approx W_p \alpha, \quad \alpha \in \Delta^{m-1}.
\label{eq:bundle_approx}
\end{equation}

This construction is exact when $p$ is affine. Importantly, restricting $\alpha$ to the simplex automatically confines solutions to $\mathrm{conv}(\{y_1, \ldots, y_m\}) \subseteq B(\bar{y}, \Delta)$, providing an implicit trust region.

\subsubsection{Multiple-Shooting Trajectory Representation}

For trajectory optimization over horizon $H$, we employ a multiple-shooting discretization \cite{bock1984,betts2010}, treating states at discrete time points as independent variables linked through dynamics constraints. At each time index $k \in \{1, \ldots, H\}$, we sample $m$ candidate state-control pairs around the current iterate $(\bar{x}_k, \bar{u}_k)$ and evaluate:
\begin{align}
W_x^{(k)} &= \begin{bmatrix} x_1^{(k)} & \cdots & x_m^{(k)} \end{bmatrix} \in \R^{n_x \times m}, \label{eq:Wx_k}\\
W_u^{(k)} &= \begin{bmatrix} u_1^{(k)} & \cdots & u_m^{(k)} \end{bmatrix} \in \R^{n_u \times m}, \label{eq:Wu_k}\\
W_f^{(k)} &= \begin{bmatrix} F(x_1^{(k)}, u_1^{(k)}) & \cdots & F(x_m^{(k)}, u_m^{(k)}) \end{bmatrix} \in \R^{n_x \times m}, \label{eq:Wf_k}\\
W_r^{(k)} &= \begin{bmatrix} r_k(x_1^{(k)}, u_1^{(k)}) & \cdots & r_k(x_m^{(k)}, u_m^{(k)}) \end{bmatrix} \in \R^{n_r \times m}, \label{eq:Wr_k}\\
W_c^{(k)} &= \begin{bmatrix} c_k(x_1^{(k)}, u_1^{(k)}) & \cdots & c_k(x_m^{(k)}, u_m^{(k)}) \end{bmatrix} \in \R^{n_c \times m}, \label{eq:Wc_k}
\end{align}
where $F$ is the discrete-time dynamics map \eqref{eq:deterministic_map}, $r_k$ is the stage cost residual, and $c_k$ captures operational constraints at step $k$. Here $c_k = (c_{\mathrm{hard},k}, c_{1,k}, \ldots, c_{J,k})$ decomposes into hard constraints (state/input bounds) and $J$ classes of soft constraints. We denote the corresponding bundle matrices as $W_{c,\mathrm{hard}}^{(k)}$ and $W_{c,j}^{(k)}$ for $j = 1, \ldots, J$.

\subsection{Constrained Bundle Subproblem}

Using the bundle matrices, we formulate a convex approximation of the nonlinear trajectory optimization problem. At iteration $\ell$, the subproblem is:
\begin{equation}
\begin{aligned}
\minimize_{\alpha, s, w, d} \quad & \sum_{k=1}^{H} \|W_r^{(k)} \alpha^{(k)}\|_2^2 + \mu_\ell \sum_{k=1}^{H-1} \left( \|s_k\|_1 + \|w_k\|_1 \right) \\
& + \sum_{j=1}^{J} \gamma_{j,\ell} \sum_{k=1}^{H-1} \|d_{k,j}\|_1 \\
\st \quad & W_f^{(k)} \alpha^{(k)} = W_x^{(k+1)} \alpha^{(k+1)} + s_k, \\
& W_{c,\mathrm{hard}}^{(k)} \alpha^{(k)} + w_k \geq 0, \\
& W_{c,j}^{(k)} \alpha^{(k)} + d_{k,j} \geq 0, \quad j = 1, \ldots, J, \\
& \alpha^{(k)} \in \Delta^{m-1}, \quad k = 1, \ldots, H, \\
& s_k \in \R^{n_x}, \; w_k \geq 0, \; d_{k,j} \geq 0,
\end{aligned}
\label{eq:convex_subproblem}
\end{equation}

where $s_k \in \R^{n_x}$ are slack variables for dynamics defects (unconstrained in sign), $w_k \in \R_{\geq 0}^{n_{\mathrm{hard}}}$ are non-negative slack variables for hard inequality constraints (state/input bounds), and $d_{k,j} \in \R_{\geq 0}^{n_j}$ are non-negative slack variables for soft constraint class $j \in \{1, \ldots, J\}$. The penalty parameter $\mu_\ell > 0$ penalizes dynamics and hard constraint violations, while $\gamma_{j,\ell} > 0$ penalizes soft constraint violations in class $j$. The use of $\ell_1$ norms for exact penalty functions follows classical approaches in constrained optimization \cite{han1977,powell1978}.

Problem \eqref{eq:convex_subproblem} is a convex optimization problem: the objective combines a convex quadratic term with convex $\ell_1$ penalties, all constraints are linear, and the feasible region is a Cartesian product of simplices and non-negative orthants.

After solving \eqref{eq:convex_subproblem}, the candidate trajectory is recovered via:
\begin{equation}
x_k^{\ell+1} = W_x^{(k)} \alpha^{(k)*}, \quad u_k^{\ell+1} = W_u^{(k)} \alpha^{(k)*}, \quad k = 1, \ldots, H.
\label{eq:recover_traj}
\end{equation}

\subsection{Adaptive Mechanisms}
\label{sec:adaptive_mechanisms}

The original TBM \cite{tracy2025trajectory} uses fixed penalty parameters throughout optimization. While effective, this requires manual tuning and lacks convergence guarantees. We introduce adaptive mechanisms for trust region and penalty management that automatically adjust based on constraint satisfaction metrics. Classical trust region methods adapt $\Delta$ based on the ratio of actual to predicted objective reduction \cite{conn2000trust}. Our adaptation rules instead use constraint violation metrics, following the philosophy of \cite{tracy2025trajectory} while introducing separate penalty classes for asymmetric constraint treatment.

\subsubsection{Constraint Violation Metrics}

We define separate violation metrics for each constraint type:

\begin{definition}[Constraint Violation Metrics]
\label{def:violations}
Given optimal slack variables $(s^*, w^*, \{d_j^*\})$ from subproblem \eqref{eq:convex_subproblem} at iteration $\ell$:
\begin{align}
\nu_{\mathrm{dyn}}^\ell &= \max_{k=1,\ldots,H-1} \|s_k^*\|_\infty, \label{eq:nu_dyn}\\
\nu_{\mathrm{hard}}^\ell &= \max_{k=1,\ldots,H-1} \|w_k^*\|_\infty, \label{eq:nu_hard}\\
\nu_j^\ell &= \sum_{k=1}^{H-1} \|d_{k,j}^*\|_1, \quad j = 1, \ldots, J. \label{eq:nu_soft}
\end{align}
\end{definition}

\subsubsection{Trust Region Adaptation}

The trust region radius $\Delta^\ell > 0$ controls the locality of the bundle approximation. We adapt it based on constraint satisfaction:

\begin{equation}
\Delta^{\ell+1} = \begin{cases}
\min(\beta_{\exp} \Delta^\ell, \Delta_{\max}) & \text{if } \nu_{\mathrm{dyn}}^\ell < \tau_{\mathrm{feas}} \text{ and } \nu_{\mathrm{hard}}^\ell < \tau_{\mathrm{feas}}, \\
\max(\beta_{\con} \Delta^\ell, \Delta_{\min}) & \text{if } \nu_{\mathrm{dyn}}^\ell > \tau_{\mathrm{viol}} \text{ or } \nu_{\mathrm{hard}}^\ell > \tau_{\mathrm{viol}}, \\
\Delta^\ell & \text{otherwise},
\end{cases}
\label{eq:trust_region_adapt}
\end{equation}
where $\beta_{\exp} > 1$ is the expansion factor, $\beta_{\con} < 1$ is the contraction factor, and $\Delta_{\min}, \Delta_{\max}$ bound the trust region radius. The thresholds satisfy $0 < \epsilon_{\mathrm{feas}} < \tau_{\mathrm{feas}} < \tau_{\mathrm{viol}}$, where $\epsilon_{\mathrm{feas}}$ is the convergence tolerance, $\tau_{\mathrm{feas}}$ triggers trust region expansion, and $\tau_{\mathrm{viol}}$ triggers contraction and penalty increase.

The rationale is as follows: when violations are small (below $\tau_{\mathrm{feas}}$), the bundle approximation is accurate over a larger region, justifying trust region expansion; when violations are large (above $\tau_{\mathrm{viol}}$), the approximation is poor, so reducing $\Delta$ forces samples closer to $z^\ell$ for better local accuracy; when violations are moderate (between $\tau_{\mathrm{feas}}$ and $\tau_{\mathrm{viol}}$), the current sampling is appropriate and the trust region radius is maintained.

\subsubsection{Penalty Adaptation}

Penalty parameters increase when violations persist above thresholds. This finite stabilization property is standard for geometrically increasing penalty sequences \cite{nocedal2006numerical}. Since $\mu$ penalizes both dynamics and hard constraint violations in the subproblem \eqref{eq:convex_subproblem}, we increase $\mu$ when either type of violation exceeds its threshold:

\begin{equation}
\mu^{\ell+1} = \begin{cases}
\min(\rho_\mu \mu^\ell, \mu_{\max}) & \text{if } \nu_{\mathrm{dyn}}^\ell > \tau_{\mathrm{viol}} \text{ or } \nu_{\mathrm{hard}}^\ell > \tau_{\mathrm{viol}}, \\
\mu^\ell & \text{otherwise}.
\end{cases} 
\label{eq:mu_adapt}
\end{equation}

Similarly, soft constraint penalties increase when their respective violations exceed thresholds:
\begin{equation}
\gamma_j^{\ell+1} = \begin{cases}
\min(\rho_\gamma \gamma_j^\ell, \gamma_{j,\max}) & \text{if } \nu_j^\ell > \tau_j, \\
\gamma_j^\ell & \text{otherwise},
\end{cases} \quad j = 1, \ldots, J,
\label{eq:gamma_adapt}
\end{equation}
where $\rho_\mu, \rho_\gamma > 1$ are increase factors, and $\tau_{\mathrm{viol}}, \tau_j$ are violation thresholds.

Algorithm~\ref{alg:adaptive_tbm} provides the complete specification.

\begin{algorithm}[t]
\caption{Adaptive Trajectory Bundle Method}
\label{alg:adaptive_tbm}
\begin{algorithmic}[1]
\Require Initial trajectory $z^0 = \{(x_k^0, u_k^0)\}_{k=1}^H$, parameters $(\Delta^0, \mu^0, \boldsymbol{\gamma}^0)$
\Ensure Converged trajectory $z^*$

\State $\ell \gets 0$

\While{not converged}
    \State \textbf{Sample:} Generate $\mathcal{Y}^\ell = \{y_1, \ldots, y_m\} \subset B(z^\ell, \Delta^\ell)$
    \State \textbf{Build Bundles:} Evaluate $W_x^{(k)}, W_u^{(k)}, W_f^{(k)}, W_r^{(k)}, W_c^{(k)}$ for $k = 1, \ldots, H$ 
    \State \textbf{Solve Sub:} Compute $(\alpha^*, s^*, w^*, \{d_j^*\})$ from \eqref{eq:convex_subproblem}
    \State \textbf{Recover:} $z^{\ell+1} \gets \{(W_x^{(k)} \alpha^{(k)*}, W_u^{(k)} \alpha^{(k)*})\}_{k=1}^H$
    
    \State \textbf{Compute Violations:} $\nu_{\mathrm{dyn}}^\ell \gets \max_k \|s_k^*\|_\infty$, $\nu_{\mathrm{hard}}^\ell \gets \max_k \|w_k^*\|_\infty$, $\nu_j^\ell \gets \sum_k \|d_{k,j}^*\|_1$
    
    \State \textbf{Adapt Trust Region:} Update $\Delta^{\ell+1}$ via \eqref{eq:trust_region_adapt}
    \State \textbf{Adapt Penalties:} Update $\mu^{\ell+1}$ via \eqref{eq:mu_adapt}, $\gamma_j^{\ell+1}$ via \eqref{eq:gamma_adapt}
    
    \State \textbf{Check Convergence:} If $\nu_{\mathrm{dyn}}^\ell < \epsilon_{\mathrm{feas}}$ and $\nu_{\mathrm{hard}}^\ell < \epsilon_{\mathrm{feas}}$ and $\|z^{\ell+1} - z^\ell\| < \epsilon_z$: \textbf{break}
    
    \State $\ell \gets \ell + 1$
\EndWhile

\State \Return $z^{\ell+1}$
\end{algorithmic}
\end{algorithm}

\section{Simulation Results}

\subsection{Implementation Details}
The proposed Adaptive TBM was evaluated on a simulated roll-to-roll system with $N=6$ motorized rollers and coupled web spans, representing a typical industrial functional film processing line.

Reference roller velocities $v_i^r(t)$ are computed from mass conservation to ensure equilibrium:
\begin{equation}
    v_i^r(t) = \frac{EA - T_{i-1}(t)}{EA - T_i(t)} v_{i-1}^r(t), \quad v_0^r = v_0(t),
    \label{eq:ref_velocity}
\end{equation}

The Adaptive TBM was implemented with prediction horizon $H=15$ steps ($0.15$s), solving the convex subproblem \eqref{eq:convex_subproblem} at each timestep. The cost residual encodes quadratic tracking:
\begin{equation}
r_k(x_k, u_k) = \begin{bmatrix} Q^{1/2}(x_k - x_k^r) \\ R^{1/2}(u_k - u_k^r) \\ S^{1/2}\Delta u_k \end{bmatrix},
\label{eq:cost_residual}
\end{equation}
with $Q = \mathrm{diag}(100 I_N, 10 I_N)$, $R = I$, and $S = 0.1I$. Hard constraints enforce state/input bounds, while soft constraints capture asymmetric tension limits with penalties $\gamma_+ = 100$ (over-tension) and $\gamma_- = 10$ (under-tension) to prioritize breakage avoidance over wrinkling.

At each time index $k$, $m = 6N + 21$ samples are generated: a deterministic stencil of $6N$ samples perturbing $(\bar{x}_k, \bar{u}_k)$ by $\pm\Delta^\ell$ along coordinate axes, plus 20 stochastic samples from $\mathcal{N}((\bar{x}_k, \bar{u}_k), \Sigma)$. The resulting convex subproblem is solved in 10--15~ms. Trust region parameters are $\beta_{\exp} = 1.5$, $\beta_{\mathrm{con}} = 0.5$, $\tau_{\mathrm{feas}} = 10^{-4}$, and $\tau_{\mathrm{viol}} = 10^{-2}$; penalties adapt with $\rho_\mu = \rho_\gamma = 2.0$.

\subsection{Tension Step Tracking}

This test evaluates tracking of a tension step change while maintaining regulation in coupled zones. Webs 1, 2, 4, 5, and 6 hold constant references of $28$, $36$, $40$, $24$, and $32\,\text{N}$, respectively. Web 3 undergoes a step from $20\,\text{N}$ to $44\,\text{N}$ at $t=0.5\,\text{s}$, with unwinding velocity fixed at $0.01\,\text{m/s}$.

Fig.~\ref{fig:tension_step_results} shows the tension response. The Adaptive TBM tracks the $24\,\text{N}$ step in Web 3 with rise time comparable to MPC while effectively rejecting coupling disturbances in downstream webs ($T_4$--$T_6$), where MPPI-based methods exhibit sustained oscillations. Fig.~\ref{fig:tension_step_inputs} confirms that TBM generates smooth torque commands, avoiding the high-frequency chatter observed in MPPI responses.

\begin{figure}[htbp]
    \centering
    \subfloat[Tension response.]{\includegraphics[width=0.48\linewidth]{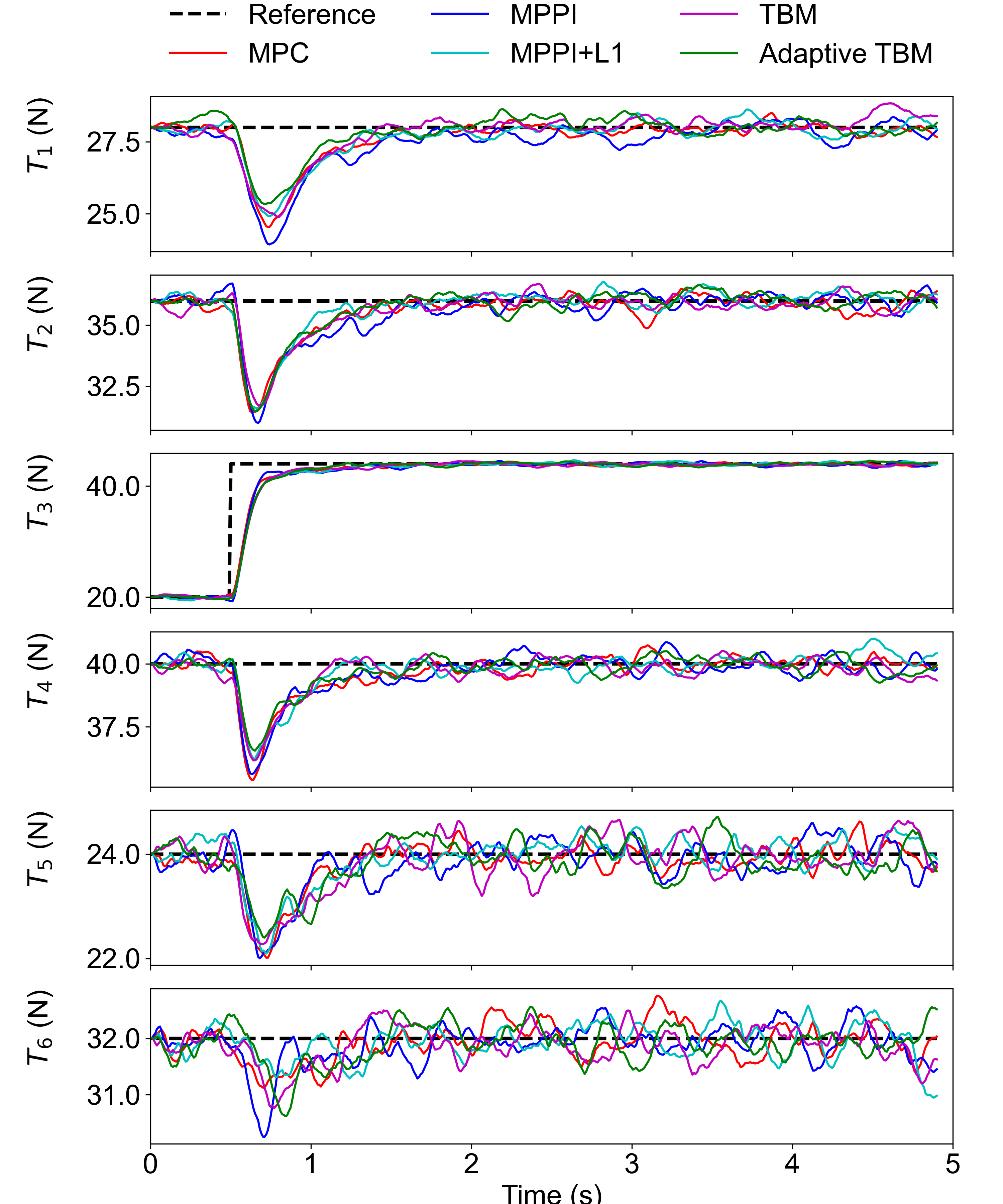}\label{fig:tension_step_results}}
    \hfill
    \subfloat[Control inputs.]{\includegraphics[width=0.48\linewidth]{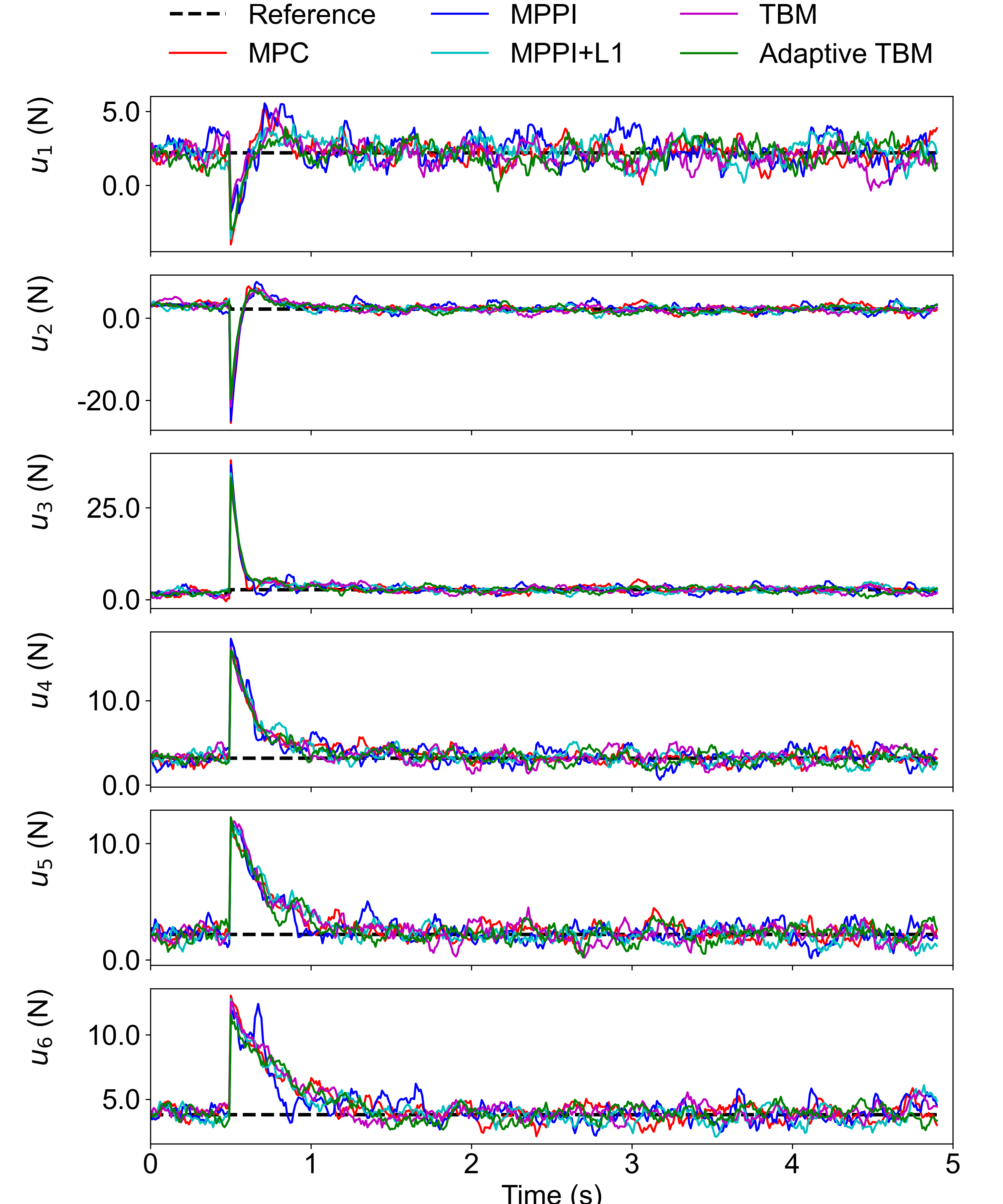}\label{fig:tension_step_inputs}}
    \caption{Step tracking performance comparison.}
    \label{fig:step_tracking}
\end{figure}

\subsection{Velocity Setpoint Change}

This test evaluates tension regulation during a velocity transient. All webs maintain reference tension of $30\,\text{N}$ while unwinding velocity increases from $0.01\,\text{m/s}$ to $0.10\,\text{m/s}$ at $t=0.5\,\text{s}$.

Fig.~\ref{fig:tension_results} shows that while all controllers experience transient tension drops due to acceleration, Adaptive TBM stabilizes downstream tensions ($T_4$--$T_6$) faster than MPPI-based baselines. Fig.~\ref{fig:control_results} confirms smooth torque generation similar to MPC.

\begin{figure}[htbp]
    \centering
    \subfloat[Tension response.]{\includegraphics[width=0.48\linewidth]{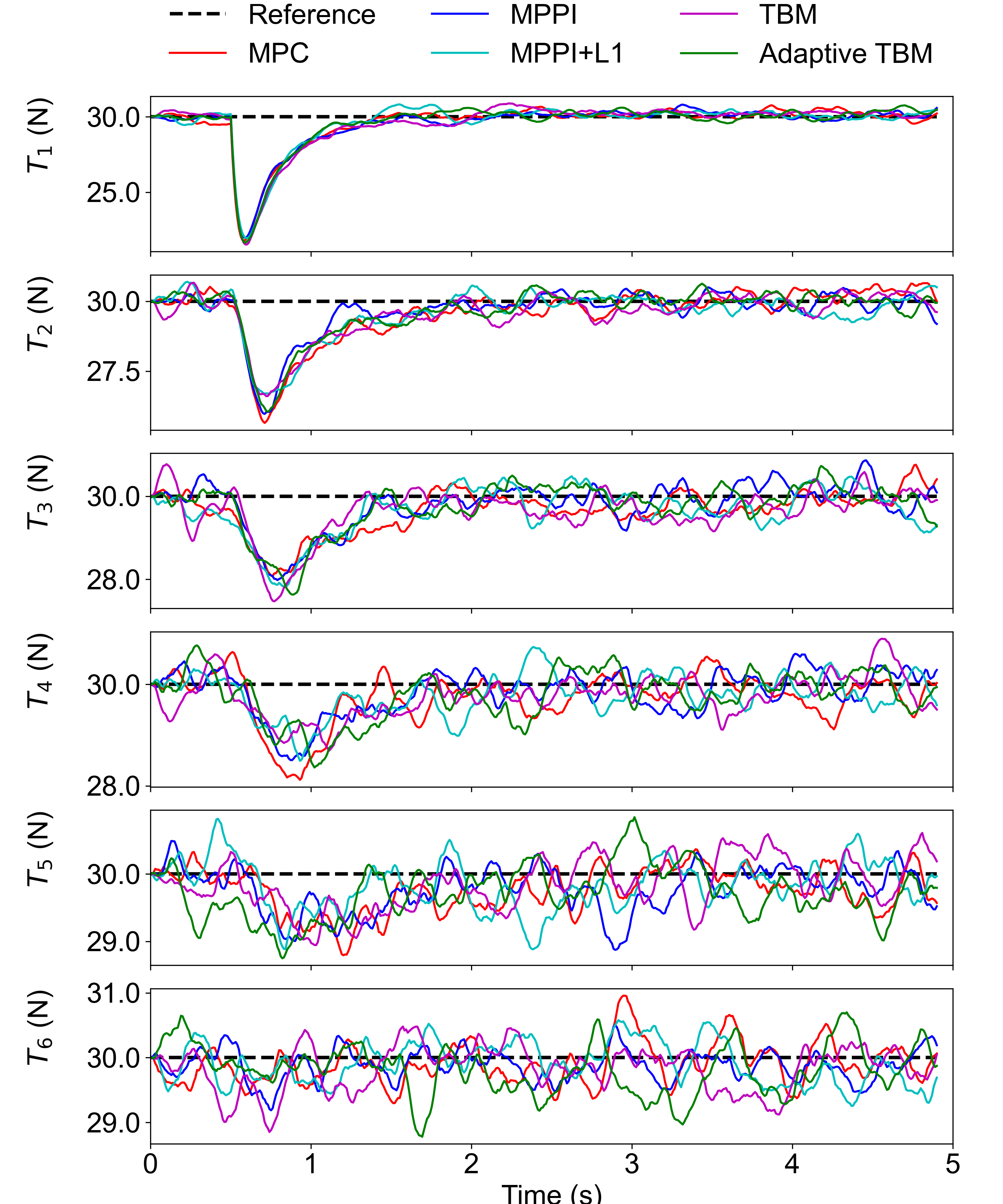}\label{fig:tension_results}}
    \hfill
    \subfloat[Control inputs.]{\includegraphics[width=0.48\linewidth]{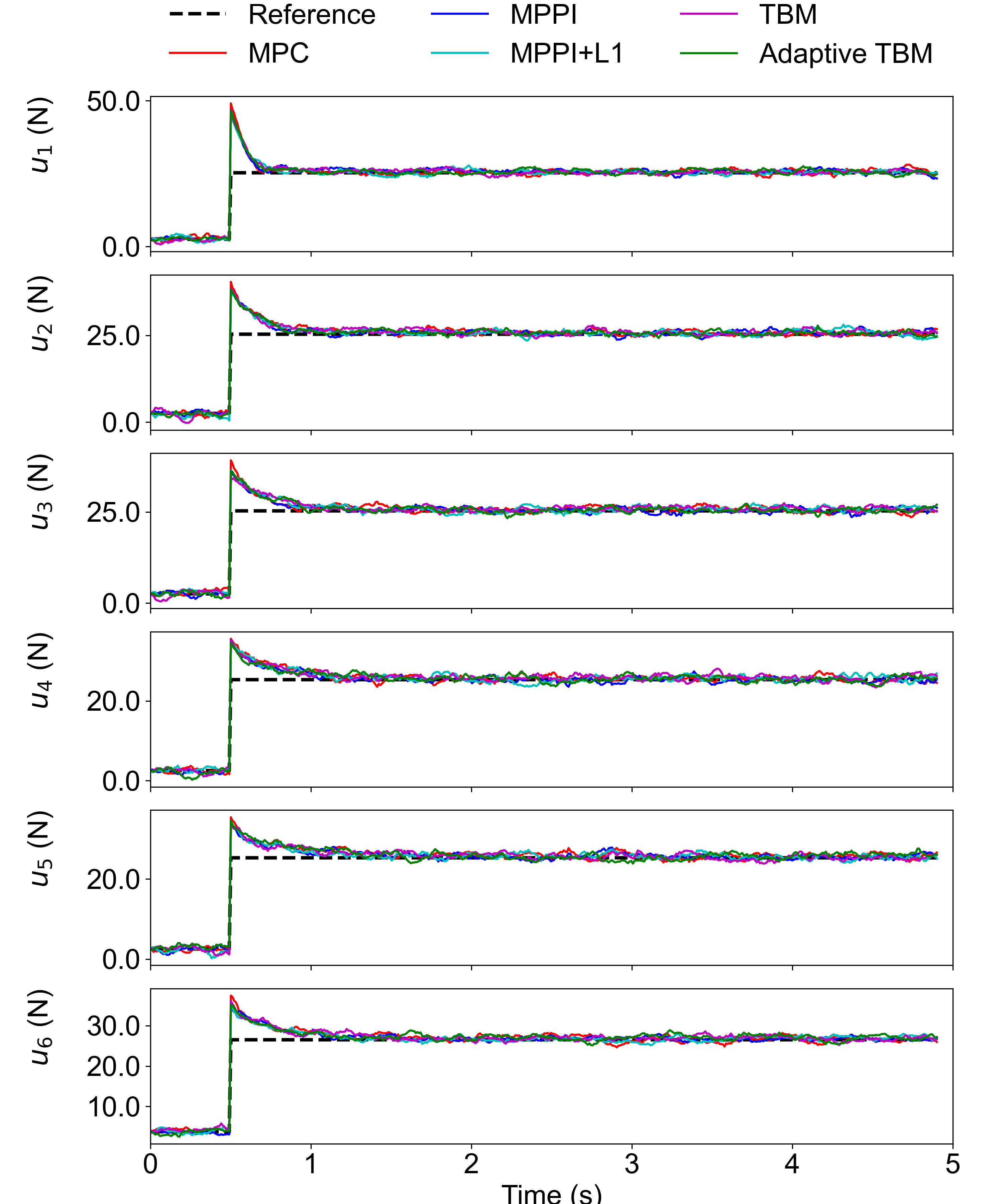}\label{fig:control_results}}
    \caption{Velocity change performance comparison.}
    \label{fig:velocity_change}
\end{figure}

Table~\ref{tab:rmse} summarizes the tension RMSE across both test scenarios. Adaptive TBM achieves the lowest RMSE in both cases, outperforming gradient-based MPC by $4.3\%$ in tension tracking and $11.1\%$ in velocity transient regulation, while maintaining smooth actuation without derivative computation.

\begin{table}[htbp]
\centering
\caption{Tension RMSE Comparison (N)}
\label{tab:rmse}
\begin{tabular}{lcc}
\hline
\textbf{Method} & \textbf{Tension Step} & \textbf{Velocity Change} \\
\hline
MPC          & 1.384 & 0.879 \\
MPPI         & 1.419 & 0.868 \\
MPPI+L1      & 1.372 & 0.815 \\
TBM          & 1.387 & 0.824 \\
Adaptive TBM & \textbf{1.324} & \textbf{0.781} \\
\hline
\end{tabular}
\end{table}

\section{Experimental Validation: R2R Dry Transfer}
\label{sec:experiment}

This section validates the proposed Adaptive TBM on an R2R mechanical dry transfer system for transferring advanced 2D materials.

\subsection{Dry Transfer Dynamics}

The R2R dry transfer process, illustrated in Fig.~\ref{fig:Dry}, peels functional 2D material from a donor substrate and deposits it onto a target substrate. The system comprises three driven rollers: one unwinding (roller 1) and two rewinding (rollers 2 and 3).

\begin{figure}[htbp]
    \centering
    \includegraphics[width=0.25\textwidth]{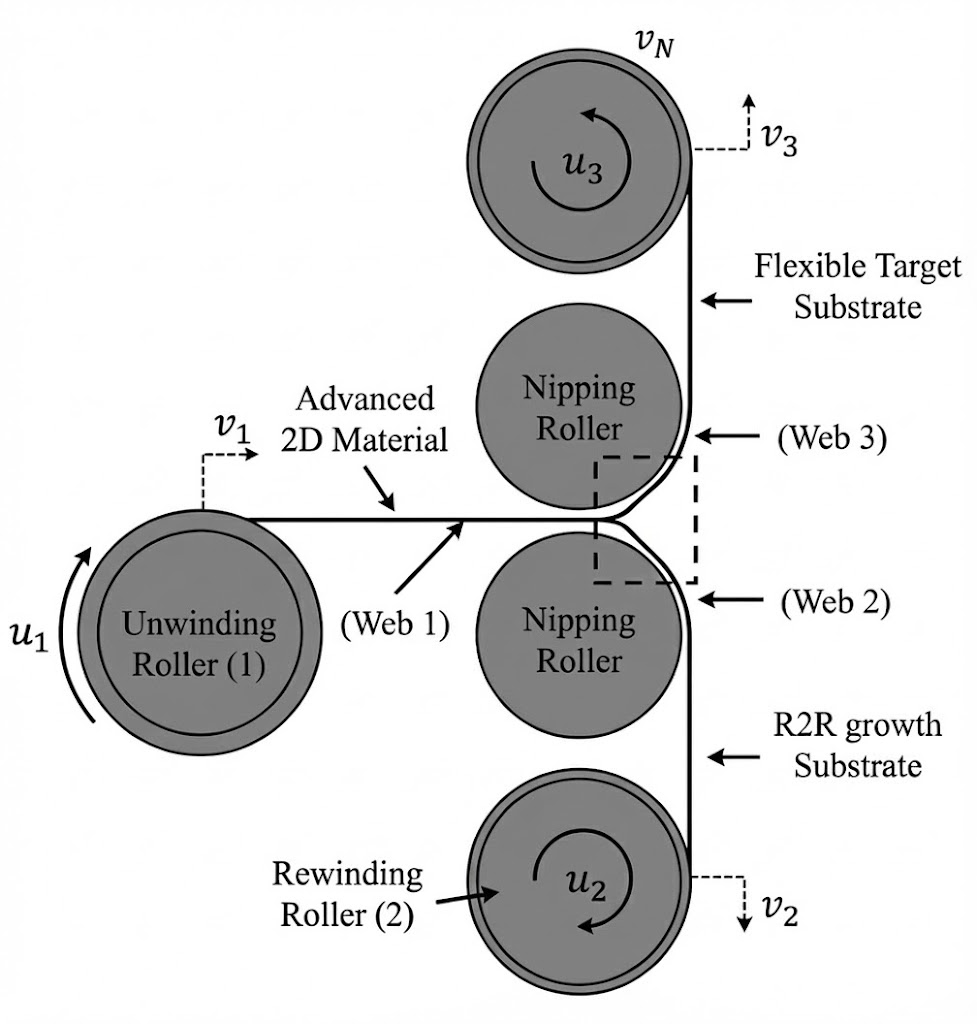}
    \caption{Illustration of the R2R peeling front.}
    \label{fig:Dry}
\end{figure}

When peeling occurs, the energy balance at the peeling front yields \cite{zhao2022dynamic,hong2023adhesion}:
\begin{equation}
T_2 + T_3 - T_1 = \sigma,
\label{eq:energy_balance}
\end{equation}
where $T_i$ denotes web tension in section $i$, and $\sigma$ combines bending and adhesion energy effects. The web velocity dynamics follow from torque balance on the rewinding rollers:
\begin{equation}
\dot{v}_i = -\frac{R_i^2}{J_i}T_i + \frac{R_i}{J_i}u_i - \frac{f_i}{R_i}v_i, \quad i = 2, 3,
\label{eq:velocity_dyn}
\end{equation}
where $R_i$, $J_i$, $f_i$ are the radius, inertia, and friction coefficient of roller $i$, and $u_i$ is the motor torque input. The unstretched web length dynamics follow from mass conservation:
\begin{equation}
\dot{l}_1 = \frac{v_1 - v_p}{1 + \varepsilon_1}, \quad \dot{l}_i = \frac{v_p}{1 + \varepsilon_1} - \frac{v_i}{1 + \varepsilon_i}, \quad i = 2, 3,
\label{eq:length_dyn}
\end{equation}
where $\varepsilon_i = T_i/(A_i E_i)$ is strain and $v_p$ is the peeling front velocity. Using the one-to-one relationship between tensions and unstretched lengths \cite{zhao2022dynamic}, the tension derivatives are:
\begin{equation}
\dot{T}_i = \frac{\partial T_i}{\partial l_1}\dot{l}_1 + \frac{\partial T_i}{\partial l_2}\dot{l}_2 + \frac{\partial T_i}{\partial l_3}\dot{l}_3, \quad i = 2, 3.
\label{eq:tension_dyn}
\end{equation}

Equations \eqref{eq:energy_balance}--\eqref{eq:tension_dyn} define the state-space representation:
\begin{equation}
\dot{x} = f(x, w, u), \quad x = [v_2, v_3, T_2, T_3]^\top, \quad u = [u_2, u_3]^\top,
\label{eq:state_space}
\end{equation}
where $w = [v_p, \sigma]^\top$ contains the exogenous inputs. Since $\sigma$ varies periodically with the material pattern and $v_p$ is a nonlinear function of system states, both are treated as bounded disturbances.

\subsection{Experimental Setup}

The experimental testbed is shown in Fig.~\ref{fig:testbed}. Trimaco masking tape is peeled from a PET transfer web, with adhesion properties representative of advanced 2D material transfer \cite{martin2025nldi}. The roller radii are $R_i = 0.0381$ m, inertias $J_2 = 0.65$ and $J_3 = 0.75$ kg$\cdot$m$^2$, and friction coefficients $f_2 = 9.2$ and $f_3 = 11.4$ N$\cdot$m$\cdot$s/rad. The web elastic modulus is $E = 2.0$ GPa with width $b = 19.3$ mm.

\begin{figure}[htbp]
    \centering
    \includegraphics[width=0.35\textwidth]{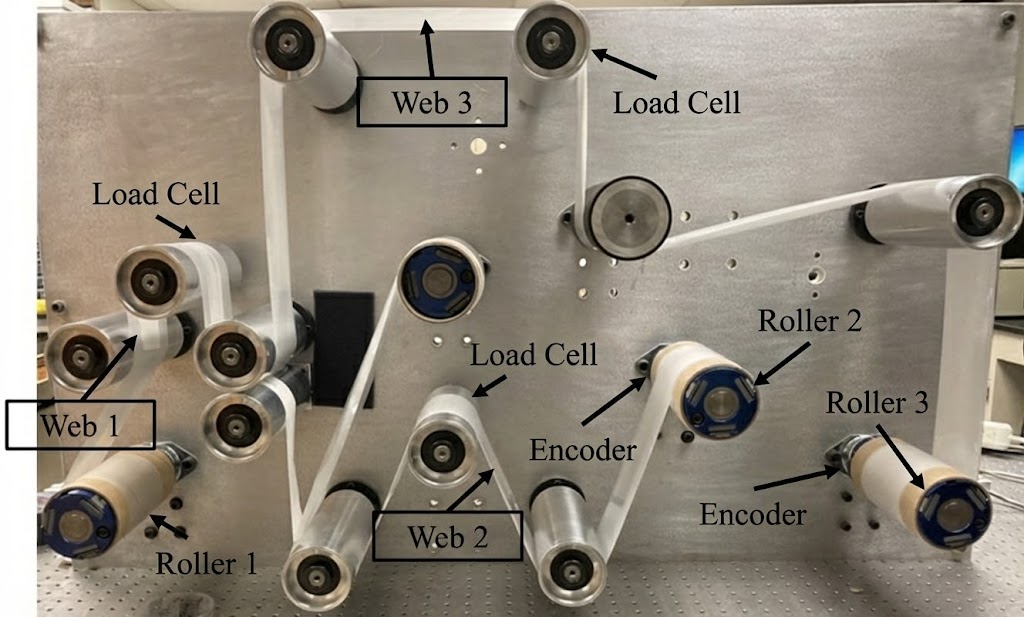}
    \caption{The experimental testbed.}
    \label{fig:testbed}
\end{figure}

The control hardware employs a National Instruments CompactRIO to collect data from quadrature encoders on rollers 2 and 3, and load cells on each web section. 

For the Adaptive TBM, we use prediction horizon $H = 15$, bundle size $m = 50$ samples, and cost weights $Q = \mathrm{diag}(1, 1, 100, 100)$, $R = 0.1 I_2$. Motor torque constraints are $|u_2| \leq 30$ N$\cdot$m and $|u_3| \leq 20$ N$\cdot$m. The trust region is initialized at $\Delta^0 = 0.5$ with bounds $[0.01, 2.0]$.

Three controllers are compared: (i) LQR \cite{martin2025nldi}, (ii) TBM with fixed penalties \cite{tracy2025trajectory}, and (iii) the proposed Adaptive TBM. The LQR controller was tuned using Bryson's rule. All methods use identical cost weights and constraints.

\subsection{Results}

Two test cases evaluate the controllers across the operating envelope: Case 1 with tensions stepping from 10~N to 18~N at $v_1 = 0.005$~m/s, and Case 2 with tensions stepping from 14~N to 22~N at $v_1 = 0.011$~m/s. The step command is applied at $t = 5$~s.

Figure~\ref{fig:exp_results} shows the tension responses for both cases. The LQR controller (blue) tracks the reference but exhibits a slower response, requiring approximately 4--5~s to fully settle. Both TBM variants achieve faster transient response than LQR. The standard TBM (red) shows moderate overshoot, particularly visible in $T_2$ for Case~2. The Adaptive TBM (purple) demonstrates improved performance, achieving faster settling with reduced overshoot through automatic trust region adjustment during transients. Across both cases, the Adaptive TBM consistently outperforms the standard TBM and LQR. 

\begin{figure}[htbp]
    \centering
    \includegraphics[width=0.5\textwidth]{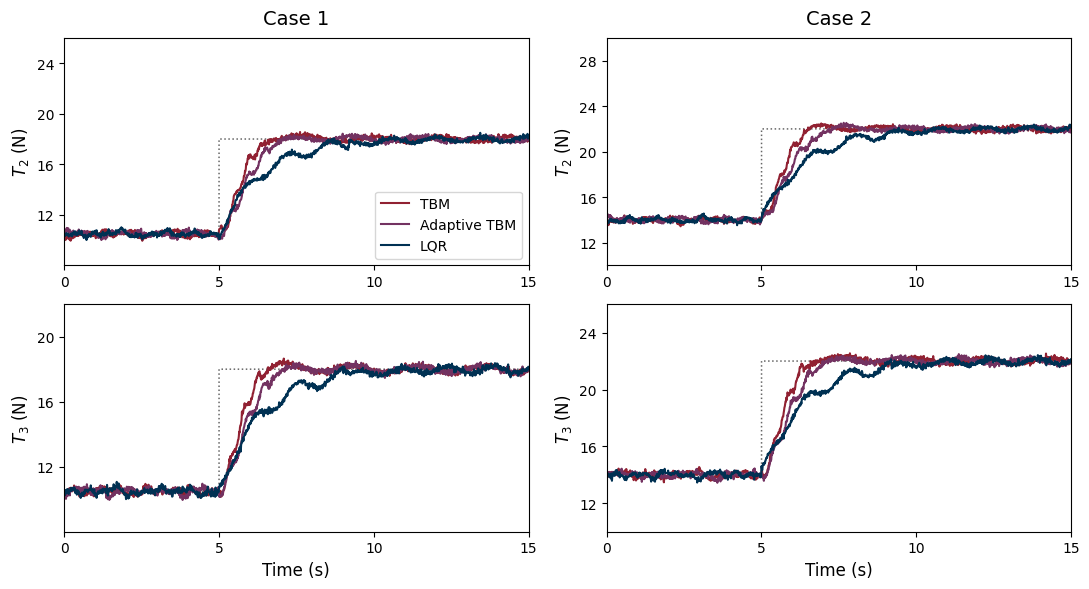}
    \caption{Experimental tension responses for Case~1  and Case~2 . Top row: $T_2$; bottom row: $T_3$. }
    \label{fig:exp_results}
\end{figure}

\section{Conclusion}

This paper developed an adaptive trajectory bundle method for constrained R2R control, combining derivative-free optimization with rigorous constraint handling. The method approximates nonlinear dynamics through bundle interpolation and employs adaptive trust region and penalty mechanisms that eliminate manual tuning. Simulations on a six-zone R2R system showed 4.3\% and 11.1\% lower tension RMSE than gradient-based MPC in step tracking and velocity transients, respectively, with smoother actuation than MPPI variants. Experimental validation on an R2R dry transfer system confirmed faster settling and reduced overshoot compared to LQR and non-adaptive TBM. The approach offers a practical solution for industrial web handling where gradients are unavailable and hard constraint satisfaction is critical.

\section*{Acknowledgments}
This work is based upon work partially supported by the National Science Foundation under Cooperative Agreement No. CMMI-2041470. Any opinions, findings and conclusions expressed in this material are those of the author(s) and do not
necessarily reflect the views of the National Science Foundation.

\bibliographystyle{IEEEtran}
\bibliography{references}

\newpage

\vspace{11pt}

\begin{IEEEbiography}[{\includegraphics[width=1.1in,height=1.5in,clip,keepaspectratio]{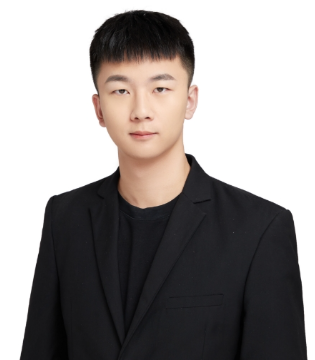}}]{Jiachen Li}
received the B.Eng. degree in mechanical engineering from the Harbin Institute of Technology, Harbin, China, in 2023. 

He is currently pursuing the Ph.D. degree in Mechanical Engineering at the University of Texas at Austin, Austin, TX, USA. His current research interests include data-driven control, energy-efficient task planning for autonomous mobile robots, and influence-function-based data attribution for dynamical systems.
\end{IEEEbiography}

\begin{IEEEbiography}
[{\includegraphics[width=1.1in,height=1.4in,clip,keepaspectratio]{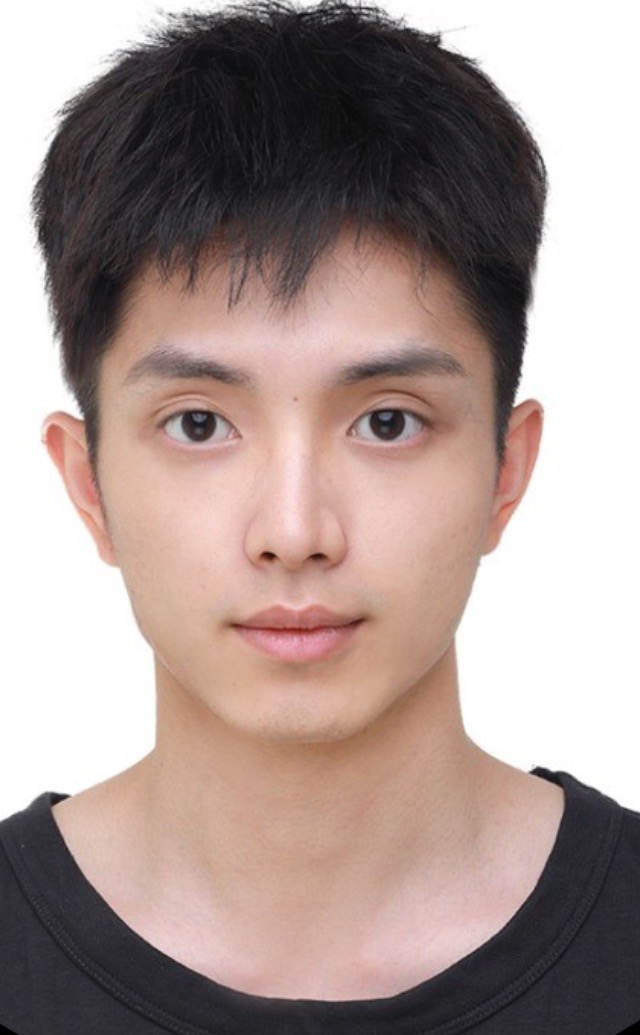}}]{Shihao Li}
received the B.S. degree in mechanical engineering from Pennsylvania State University, State College, PA, United States, in 2023.

He is currently working toward the Ph.D. degree in Mechanical Engineering at the University of Texas at Austin, Austin, TX, USA. His research interests include data-driven control, safe learning-based control, and reinforcement learning with data attribution for dynamical systems.
\end{IEEEbiography}

\begin{IEEEbiography}[{\includegraphics[width=1.1in,height=1.5in,clip,keepaspectratio]{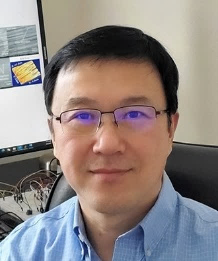}}]{Wei Li}
Wei Li is a professor of mechanical engineering at the Walker Department of Mechanical Engineering at The University of Texas at Austin. He received his B.S. degree in precision instrument and mechanology from Tsinghua University and Ph.D. degree in mechanical engineering from the University of Michigan, Ann Arbor. His research interests are in the area of nano and biomaterials processing and manufacturing, including the growth and transfer of two-dimensional materials, biomedical device design and manufacturing, and monitoring, diagnosis, and control of manufacturing processes. Dr. Li is a recipient of a CAREER award from the National Science Foundation and a Presidential Early Career Award for Scientists and Engineers (PECASE).
\end{IEEEbiography}

\begin{IEEEbiography}[{\includegraphics[width=1.1in,height=1.5in,clip,keepaspectratio]{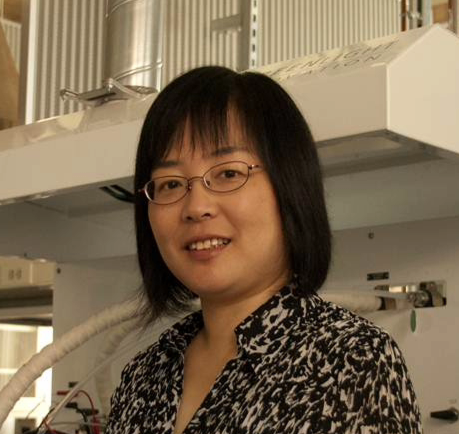}}]{Dongmei “Maggie” Chen}
is the J. Mike Walker Professor in the Walker Department of Mechanical Engineering at The University of Texas at Austin, where she has been a faculty member since January 2009. Prior to that, she was a Senior Control Algorithms Engineer at the General Motors Fuel Cell Activities Center.

Her teaching and research focus on automatic control and dynamic systems. Dr. Chen has authored or co-authored more than 150 technical papers and holds six patents. She received the National Science Foundation (NSF) CAREER Award in 2011 and the IEEE PES Prize Paper Award in 2016.

She received her Ph.D. in Mechanical Engineering from the University of Michigan and her B.S. in Precision Instruments and Mechanology from Tsinghua University. She is a member of ASME, IEEE, and SPE.
\end{IEEEbiography}

\vfill

\newpage
\onecolumn
\appendices
\section{Convergence Analysis}
\label{sec:convergence}

This section establishes convergence guarantees for Algorithm~\ref{alg:adaptive_tbm}. We prove that the adaptive mechanisms drive constraint violations below specified thresholds in finite iterations and that the algorithm converges to stationary points of the penalized objective, which coincide with stationary points of the original constrained R2R control problem under sufficient penalty magnitudes. The proof structure follows classical trust region convergence theory \cite{conn2000trust}, adapted to the penalty formulation and always-accept strategy of \cite{tracy2025trajectory}.

\subsection{Penalized Objective and Stationarity}

We first define the penalized objective function that underlies the convergence analysis.

\begin{definition}[Penalized Objective]
\label{def:penalty}
For trajectory $z = \{(x_k, u_k)\}_{k=1}^{H}$ and penalty parameters $(\mu, \boldsymbol{\gamma}) = (\mu, \gamma_1, \ldots, \gamma_J)$, define:
\begin{equation}
\begin{aligned}
\phi_{\mu, \boldsymbol{\gamma}}(z) = \; & \sum_{k=1}^{H} \|r_k(x_k, u_k)\|_2^2 \\
& + \mu \sum_{k=1}^{H-1} \|x_{k+1} - F(x_k, u_k)\|_1 \\
& + \mu \sum_{k=1}^{H-1} \|[c_{\mathrm{hard},k}(x_k, u_k)]_-\|_1 \\
& + \sum_{j=1}^{J} \gamma_j \sum_{k=1}^{H-1} \|[c_{j,k}(x_k, u_k)]_-\|_1,
\end{aligned}
\label{eq:penalty_function}
\end{equation}
where $[v]_- = \max(0, -v)$ is applied element-wise, extracting constraint violations.
\end{definition}

The penalized objective $\phi_{\mu, \boldsymbol{\gamma}}$ is nonsmooth due to the $\ell_1$ norms and $[\cdot]_-$ operators. Classical gradient-based stationarity concepts do not apply directly. We adopt an algorithmic stationarity definition appropriate for derivative-free methods.

\begin{definition}[Algorithmic Stationarity]
\label{def:stationarity}
A trajectory $z^*$ is \emph{$\epsilon$-stationary} for $\phi_{\mu, \boldsymbol{\gamma}}$ if for all sufficiently small $\Delta > 0$ and any sample set $\mathcal{Y} \subset B(z^*, \Delta)$ satisfying Assumption~\ref{ass:sampling}, solving the subproblem \eqref{eq:convex_subproblem} yields a recovered trajectory $z^+$ with:
\begin{equation}
\|z^+ - z^*\| \leq \epsilon.
\end{equation}
A trajectory is \emph{stationary} if it is $\epsilon$-stationary for all $\epsilon > 0$.
\end{definition}

\begin{remark}[Relationship to Classical Stationarity]
\label{rem:stationarity}
Definition~\ref{def:stationarity} captures the practical notion that no local convex combination of nearby samples significantly improves the objective. For Lipschitz functions, this algorithmic stationarity implies Clarke stationarity under mild regularity conditions \cite[Chapter~8]{conn2009derivative}. Specifically, if $z^*$ is stationary in the sense of Definition~\ref{def:stationarity} and $\phi_{\mu,\boldsymbol{\gamma}}$ is locally Lipschitz, then $0 \in \partial_C \phi_{\mu,\boldsymbol{\gamma}}(z^*)$ where $\partial_C$ denotes the Clarke subdifferential. The converse holds when the sampling provides sufficient directional coverage (Assumption~\ref{ass:sampling}(iii)).
\end{remark}

\subsection{Standing Assumptions}

\begin{assumption}[Lipschitz Continuity]
\label{ass:lipschitz}
The stage cost residual $r_k: \R^{n_x} \times \R^{n_u} \to \R^{n_r}$, dynamics map $F: \R^{n_x} \times \R^{n_u} \to \R^{n_x}$, and constraint functions $c_k: \R^{n_x} \times \R^{n_u} \to \R^{n_c}$ are Lipschitz continuous with constants $L_r$, $L_F$, and $L_c$ respectively. That is, for all $(x, u), (x', u')$:
\begin{align}
\|r_k(x, u) - r_k(x', u')\| &\leq L_r \|(x, u) - (x', u')\|, \\
\|F(x, u) - F(x', u')\| &\leq L_F \|(x, u) - (x', u')\|, \\
\|c_k(x, u) - c_k(x', u')\| &\leq L_c \|(x, u) - (x', u')\|.
\end{align}
\end{assumption}

\begin{assumption}[Bounded Level Sets]
\label{ass:bounded}
For any fixed penalties $(\mu, \boldsymbol{\gamma})$ with $\mu \leq \mu_{\max}$ and $\gamma_j \leq \gamma_{j,\max}$, the level set
\begin{equation}
\mathcal{L}(\mu, \boldsymbol{\gamma}) = \left\{ z : \phi_{\mu, \boldsymbol{\gamma}}(z) \leq \phi_{\mu, \boldsymbol{\gamma}}(z^0) \right\}
\end{equation}
is compact. Furthermore, the cost residuals are uniformly bounded on this level set: there exists $R > 0$ such that $\|r_k(x_k, u_k)\| \leq R$ for all $z \in \mathcal{L}(\mu, \boldsymbol{\gamma})$ and all $k \in \{1, \ldots, H\}$.
\end{assumption}

\begin{assumption}[Controlled Sampling]
\label{ass:sampling}
At each iteration $\ell$, the sampling procedure generates $\mathcal{Y}^\ell = \{y_1, \ldots, y_m\} \subset B(z^\ell, \Delta^\ell)$ satisfying:
\begin{enumerate}
    \item[(i)] \textbf{Containment:} $y_i \in B(z^\ell, \Delta^\ell)$ for all $i \in \{1, \ldots, m\}$;
    \item[(ii)] \textbf{Inclusion of iterate:} $z^\ell \in \mathcal{Y}^\ell$;
    \item[(iii)] \textbf{Directional coverage:} There exists $\kappa > 0$ (independent of $\ell$) such that for any unit direction $d$ with $\|d\| = 1$, there exists a sample $y_i \in \mathcal{Y}^\ell$ satisfying $\langle y_i - z^\ell, d \rangle \geq \kappa \Delta^\ell$.
\end{enumerate}
\end{assumption}

\begin{remark}[Verification of Assumption~\ref{ass:sampling}]
\label{rem:sampling_verification}
The sampling strategy described in Section~V satisfies Assumption~\ref{ass:sampling}:
\begin{itemize}
    \item Condition (i) holds by construction since all samples are generated within $B(z^\ell, \Delta^\ell)$.
    \item Condition (ii) is satisfied by explicitly including $z^\ell$ in the sample set.
    \item Condition (iii) is satisfied with $\kappa = 1/\sqrt{n_x + n_u}$ by the coordinate perturbation samples $z^\ell \pm \Delta^\ell e_j$ for each coordinate direction $e_j$. For any unit direction $d = \sum_j d_j e_j$, at least one coordinate satisfies $|d_j| \geq 1/\sqrt{n_x + n_u}$. The corresponding perturbation $z^\ell + \mathrm{sign}(d_j) \Delta^\ell e_j$ satisfies $\langle y - z^\ell, d \rangle = |d_j| \Delta^\ell \geq \kappa \Delta^\ell$.
\end{itemize}
\end{remark}

\begin{assumption}[Parameter Compatibility]
\label{ass:parameter}
The algorithm parameters satisfy the compatibility condition:
\begin{equation}
\Delta_{\min} \leq \bar{\Delta} \triangleq \frac{\mu_{\max} \tau_{\mathrm{viol}}}{16 C_{\mathrm{approx}} L_\phi(\mu_{\max}, \boldsymbol{\gamma}_{\max})},
\label{eq:parameter_compat}
\end{equation}
where $C_{\mathrm{approx}} = 2$ is the approximation constant from Lemma~\ref{lem:bundle_accuracy} and $L_\phi(\mu_{\max}, \boldsymbol{\gamma}_{\max})$ is the Lipschitz constant of the penalized objective at maximum penalty values (Lemma~\ref{lem:phi_lipschitz}).
\end{assumption}

\begin{remark}[Practical Parameter Selection]
\label{rem:parameter_selection}
Assumption~\ref{ass:parameter} ensures that the trust region can contract sufficiently to reach the critical radius $\bar{\Delta}$ required for Lemma~\ref{lem:feasibility}. In practice, $L_\phi$ may not be known exactly. However, the assumption is easily satisfied by choosing $\Delta_{\min}$ conservatively small (e.g., $\Delta_{\min} = 10^{-8}$). The algorithm will adaptively contract the trust region as needed, and the convergence guarantees hold as long as $\Delta_{\min}$ is below the (possibly unknown) threshold $\bar{\Delta}$. For the R2R system parameters in Section~V, we estimate $L_\phi \approx 10^4$ and verify that $\Delta_{\min} = 0.01$ satisfies \eqref{eq:parameter_compat}.
\end{remark}

\begin{assumption}[Exact Penalty Threshold]
\label{ass:exact_penalty}
The penalty bounds $\mu_{\max}$ and $\gamma_{j,\max}$ are chosen large enough such that the exact penalty property holds: any local minimizer of the original constrained problem that satisfies a constraint qualification is also a local minimizer of $\phi_{\mu, \boldsymbol{\gamma}}$ for $\mu \geq \mu_{\max}$ and $\gamma_j \geq \gamma_{j,\max}$.
\end{assumption}

\begin{remark}
Assumption~\ref{ass:exact_penalty} is standard for $\ell_1$ exact penalty methods \cite[Theorem~17.3]{nocedal2006numerical}. The required threshold depends on the Lagrange multipliers at optimal solutions, which are typically bounded for well-posed control problems. In practice, $\mu_{\max} = 10^6$ is sufficiently large for most R2R applications.
\end{remark}

\subsection{Lipschitz Properties}

\begin{lemma}[Lipschitz Constant of Penalized Objective]
\label{lem:phi_lipschitz}
Under Assumptions~\ref{ass:lipschitz} and \ref{ass:bounded}, the penalized objective $\phi_{\mu,\boldsymbol{\gamma}}$ is Lipschitz continuous on $\mathcal{L}(\mu, \boldsymbol{\gamma})$ with constant:
\begin{equation}
L_\phi(\mu, \boldsymbol{\gamma}) = 2HRL_r + (H-1)\mu(1 + L_F + L_c) + (H-1)L_c\sum_{j=1}^J \gamma_j,
\label{eq:L_phi}
\end{equation}
where $H$ is the prediction horizon and $R$ is the bound on cost residuals from Assumption~\ref{ass:bounded}.
\end{lemma}

\begin{proof}
The proof follows from standard Lipschitz composition rules. The $\ell_1$ norm is Lipschitz with constant 1, and the $[\cdot]_-$ operator satisfies $|[a]_- - [b]_-| \leq |a - b|$, hence is also Lipschitz with constant 1.

\textbf{Cost term:} For the quadratic cost $\|r_k(x_k, u_k)\|_2^2$, we use the identity:
\begin{equation}
|\|a\|^2 - \|b\|^2| = |\|a\| + \|b\|| \cdot |\|a\| - \|b\|| \leq (\|a\| + \|b\|) \|a - b\|.
\end{equation}
With $\|r_k\| \leq R$ on $\mathcal{L}(\mu, \boldsymbol{\gamma})$:
\begin{equation}
\left| \|r_k(x_k, u_k)\|^2 - \|r_k(x_k', u_k')\|^2 \right| \leq 2R L_r \|(x_k, u_k) - (x_k', u_k')\|.
\end{equation}
Summing over $k = 1, \ldots, H$ and using $\|(x_k, u_k) - (x_k', u_k')\| \leq \|z - z'\|$ contributes $2HRL_r$ to the Lipschitz constant.

\textbf{Dynamics penalty term:} For $\|x_{k+1} - F(x_k, u_k)\|_1$:
\begin{equation}
\begin{aligned}
&\left| \|x_{k+1} - F(x_k, u_k)\|_1 - \|x'_{k+1} - F(x'_k, u'_k)\|_1 \right| \\
&\quad \leq \|(x_{k+1} - x'_{k+1}) - (F(x_k, u_k) - F(x'_k, u'_k))\|_1 \\
&\quad \leq \|x_{k+1} - x'_{k+1}\| + L_F \|(x_k, u_k) - (x'_k, u'_k)\| \\
&\quad \leq (1 + L_F) \|z - z'\|.
\end{aligned}
\end{equation}
Multiplying by $\mu$ and summing over $k = 1, \ldots, H-1$ contributes $(H-1)\mu(1 + L_F)$.

\textbf{Hard constraint penalty:} For $\|[c_{\mathrm{hard},k}(x_k, u_k)]_-\|_1$:
\begin{equation}
\left| \|[c_{\mathrm{hard},k}(x,u)]_-\|_1 - \|[c_{\mathrm{hard},k}(x',u')]_-\|_1 \right| \leq L_c \|(x,u) - (x',u')\|.
\end{equation}
Multiplying by $\mu$ and summing contributes $(H-1)\mu L_c$.

\textbf{Soft constraint penalties:} Similarly, each class $j$ contributes $(H-1)\gamma_j L_c$.

\textbf{Total:} Combining all terms:
\begin{equation}
L_\phi = 2HRL_r + (H-1)\mu(1 + L_F) + (H-1)\mu L_c + (H-1)L_c\sum_{j=1}^J \gamma_j,
\end{equation}
which simplifies to \eqref{eq:L_phi}.
\end{proof}

\subsection{Penalty Stabilization}

\begin{lemma}[Penalty Stabilization]
\label{lem:stabilization}
Under Algorithm~\ref{alg:adaptive_tbm}, there exists $K^* < \infty$ such that $\mu^\ell = \bar{\mu}$ and $\gamma_j^\ell = \bar{\gamma}_j$ for all $\ell \geq K^*$, where $\bar{\mu} \in \{\mu^0, \rho_\mu \mu^0, \rho_\mu^2 \mu^0, \ldots\} \cap (0, \mu_{\max}]$ and similarly $\bar{\gamma}_j \in \{\gamma_j^0, \rho_\gamma \gamma_j^0, \ldots\} \cap (0, \gamma_{j,\max}]$. Furthermore:
\begin{equation}
K^* \leq \left\lceil \log_{\rho_\mu}\left(\frac{\mu_{\max}}{\mu^0}\right) \right\rceil + \sum_{j=1}^J \left\lceil \log_{\rho_\gamma}\left(\frac{\gamma_{j,\max}}{\gamma_j^0}\right) \right\rceil.
\label{eq:K_star_bound}
\end{equation}
\end{lemma}

\begin{proof}
By the penalty adaptation rules \eqref{eq:mu_adapt}--\eqref{eq:gamma_adapt}, penalties increase geometrically by factors $\rho_\mu > 1$ or $\rho_\gamma > 1$ when violations exceed thresholds, and are capped at $\mu_{\max}, \gamma_{j,\max}$.

For $\mu$: starting from $\mu^0$, after $n$ increases we have $\mu = \rho_\mu^n \mu^0$. The maximum number of increases before reaching $\mu_{\max}$ satisfies:
\begin{equation}
\rho_\mu^{n_\mu} \mu^0 \leq \mu_{\max} < \rho_\mu^{n_\mu + 1} \mu^0,
\end{equation}
giving $n_\mu = \lfloor \log_{\rho_\mu}(\mu_{\max}/\mu^0) \rfloor \leq \lceil \log_{\rho_\mu}(\mu_{\max}/\mu^0) \rceil$.

Similarly, for each $\gamma_j$, at most $n_j = \lceil \log_{\rho_\gamma}(\gamma_{j,\max}/\gamma_j^0) \rceil$ increases occur.

Since at most one penalty increase per parameter occurs per iteration, and once a penalty reaches its maximum it remains constant:
\begin{equation}
K^* \leq n_\mu + \sum_{j=1}^J n_j,
\end{equation}
which is \eqref{eq:K_star_bound}. After iteration $K^*$, all penalties have stabilized at their final values $(\bar{\mu}, \bar{\boldsymbol{\gamma}})$.
\end{proof}

\subsection{Bundle Approximation Properties}

\begin{lemma}[Bundle Approximation Accuracy]
\label{lem:bundle_accuracy}
Let $z^+ = \{(W_x^{(k)} \alpha^{(k)*}, W_u^{(k)} \alpha^{(k)*})\}_{k=1}^H$ be the trajectory recovered from subproblem \eqref{eq:convex_subproblem}. Under Assumptions~\ref{ass:lipschitz}--\ref{ass:sampling}, there exists a constant $C_{\mathrm{approx}} = 2$ such that:
\begin{equation}
\left| \phi_{\mu^\ell, \boldsymbol{\gamma}^\ell}(z^+) - J_{\mathrm{sub}}(\alpha^*, s^*, w^*, \{d_j^*\}) \right| \leq C_{\mathrm{approx}} L_\phi(\mu^\ell, \boldsymbol{\gamma}^\ell) \Delta^\ell,
\label{eq:approx_bound}
\end{equation}
where $J_{\mathrm{sub}}$ denotes the optimal subproblem objective value.
\end{lemma}

\begin{proof}
We bound the approximation error for each component of the objective separately.

\textbf{Step 1: Cost term approximation error.}

Let $(x_k^+, u_k^+) = (W_x^{(k)}\alpha^{(k)*}, W_u^{(k)}\alpha^{(k)*})$ denote the recovered state-control pair at time $k$. Since samples satisfy $\|(x_i^{(k)}, u_i^{(k)}) - (\bar{x}_k, \bar{u}_k)\| \leq \Delta^\ell$ for all $i$, and $\alpha^{(k)*} \in \Delta^{m-1}$, the recovered point is a convex combination:
\begin{equation}
(x_k^+, u_k^+) = \sum_{i=1}^m \alpha_i^{(k)*} (x_i^{(k)}, u_i^{(k)}).
\end{equation}
By convexity of the ball $B((\bar{x}_k, \bar{u}_k), \Delta^\ell)$:
\begin{equation}
\|(x_k^+, u_k^+) - (\bar{x}_k, \bar{u}_k)\| \leq \Delta^\ell.
\label{eq:recovered_in_ball}
\end{equation}

The bundle stores sampled cost values: $W_r^{(k)} = [r_k(x_1^{(k)}, u_1^{(k)}), \ldots, r_k(x_m^{(k)}, u_m^{(k)})]$. The interpolated cost is:
\begin{equation}
W_r^{(k)} \alpha^{(k)*} = \sum_{i=1}^m \alpha_i^{(k)*} r_k(x_i^{(k)}, u_i^{(k)}).
\end{equation}

This differs from the true cost $r_k(x_k^+, u_k^+)$ because $r_k$ is nonlinear. We bound the error using the triangle inequality.

By Lipschitz continuity of $r_k$:
\begin{equation}
\|r_k(x_i^{(k)}, u_i^{(k)}) - r_k(\bar{x}_k, \bar{u}_k)\| \leq L_r \Delta^\ell \quad \forall i.
\end{equation}

Therefore:
\begin{equation}
\|W_r^{(k)} \alpha^{(k)*} - r_k(\bar{x}_k, \bar{u}_k)\| = \left\| \sum_i \alpha_i^{(k)*} (r_k(x_i^{(k)}, u_i^{(k)}) - r_k(\bar{x}_k, \bar{u}_k)) \right\| \leq L_r \Delta^\ell.
\end{equation}

Similarly, by \eqref{eq:recovered_in_ball}:
\begin{equation}
\|r_k(x_k^+, u_k^+) - r_k(\bar{x}_k, \bar{u}_k)\| \leq L_r \Delta^\ell.
\end{equation}

By the triangle inequality:
\begin{equation}
\|W_r^{(k)} \alpha^{(k)*} - r_k(x_k^+, u_k^+)\| \leq 2L_r \Delta^\ell.
\label{eq:cost_interp_error}
\end{equation}

For the squared norm, using $|\|a\|^2 - \|b\|^2| \leq (\|a\| + \|b\|)\|a - b\|$ with $\|r_k(\cdot)\| \leq R$:
\begin{equation}
\left| \|W_r^{(k)} \alpha^{(k)*}\|_2^2 - \|r_k(x_k^+, u_k^+)\|_2^2 \right| \leq 2R \cdot 2L_r \Delta^\ell = 4R L_r \Delta^\ell.
\end{equation}

Summing over $k = 1, \ldots, H$:
\begin{equation}
\left| \sum_{k=1}^H \|W_r^{(k)} \alpha^{(k)*}\|_2^2 - \sum_{k=1}^H \|r_k(x_k^+, u_k^+)\|_2^2 \right| \leq 4HR L_r \Delta^\ell.
\label{eq:cost_total_error}
\end{equation}

\textbf{Step 2: Dynamics constraint approximation error.}

The subproblem enforces $W_f^{(k)} \alpha^{(k)*} = W_x^{(k+1)} \alpha^{(k+1)*} + s_k^*$. Rearranging:
\begin{equation}
s_k^* = W_f^{(k)} \alpha^{(k)*} - W_x^{(k+1)} \alpha^{(k+1)*} = W_f^{(k)} \alpha^{(k)*} - x_{k+1}^+.
\end{equation}

The true dynamics defect at the recovered trajectory is $x_{k+1}^+ - F(x_k^+, u_k^+)$. We need to relate $\|s_k^*\|_1$ to the true defect.

By the same interpolation argument as above:
\begin{equation}
\|W_f^{(k)} \alpha^{(k)*} - F(\bar{x}_k, \bar{u}_k)\| \leq L_F \Delta^\ell,
\end{equation}
and
\begin{equation}
\|F(x_k^+, u_k^+) - F(\bar{x}_k, \bar{u}_k)\| \leq L_F \Delta^\ell.
\end{equation}

Therefore:
\begin{equation}
\|W_f^{(k)} \alpha^{(k)*} - F(x_k^+, u_k^+)\| \leq 2L_F \Delta^\ell.
\label{eq:dynamics_interp_error}
\end{equation}

The slack $s_k^* = W_f^{(k)} \alpha^{(k)*} - x_{k+1}^+$ and the true defect $\delta_k = x_{k+1}^+ - F(x_k^+, u_k^+)$ satisfy:
\begin{equation}
s_k^* = W_f^{(k)} \alpha^{(k)*} - F(x_k^+, u_k^+) - \delta_k.
\end{equation}

Thus:
\begin{equation}
\left| \|s_k^*\|_1 - \|\delta_k\|_1 \right| \leq \|s_k^* + \delta_k\|_1 = \|W_f^{(k)} \alpha^{(k)*} - F(x_k^+, u_k^+)\|_1 \leq 2L_F \Delta^\ell.
\end{equation}

Summing over $k = 1, \ldots, H-1$:
\begin{equation}
\left| \sum_{k=1}^{H-1} \|s_k^*\|_1 - \sum_{k=1}^{H-1} \|x_{k+1}^+ - F(x_k^+, u_k^+)\|_1 \right| \leq 2(H-1) L_F \Delta^\ell.
\label{eq:dynamics_total_error}
\end{equation}

\textbf{Step 3: Hard constraint approximation error.}

For hard constraints, the subproblem enforces $W_{c,\mathrm{hard}}^{(k)} \alpha^{(k)*} + w_k^* \geq 0$ with $w_k^* \geq 0$. The optimal $w_k^*$ is the minimum non-negative slack:
\begin{equation}
w_k^* = [-W_{c,\mathrm{hard}}^{(k)} \alpha^{(k)*}]_+,
\end{equation}
where $[v]_+ = \max(0, v)$ element-wise.

By interpolation:
\begin{equation}
\|W_{c,\mathrm{hard}}^{(k)} \alpha^{(k)*} - c_{\mathrm{hard},k}(x_k^+, u_k^+)\| \leq 2L_c \Delta^\ell.
\end{equation}

The true hard constraint violation is $[-c_{\mathrm{hard},k}(x_k^+, u_k^+)]_+ = [c_{\mathrm{hard},k}(x_k^+, u_k^+)]_-$. Since both $[\cdot]_+$ and negation are Lipschitz with constant 1:
\begin{equation}
\left| \|w_k^*\|_1 - \|[c_{\mathrm{hard},k}(x_k^+, u_k^+)]_-\|_1 \right| \leq 2L_c \Delta^\ell.
\end{equation}

Summing over $k$:
\begin{equation}
\left| \sum_{k=1}^{H-1} \|w_k^*\|_1 - \sum_{k=1}^{H-1} \|[c_{\mathrm{hard},k}(x_k^+, u_k^+)]_-\|_1 \right| \leq 2(H-1) L_c \Delta^\ell.
\label{eq:hard_total_error}
\end{equation}

\textbf{Step 4: Soft constraint approximation error.}

By identical reasoning, for each soft constraint class $j$:
\begin{equation}
\left| \sum_{k=1}^{H-1} \|d_{k,j}^*\|_1 - \sum_{k=1}^{H-1} \|[c_{j,k}(x_k^+, u_k^+)]_-\|_1 \right| \leq 2(H-1) L_c \Delta^\ell.
\label{eq:soft_total_error}
\end{equation}

\textbf{Step 5: Combined error bound.}

The subproblem objective is:
\begin{equation}
J_{\mathrm{sub}} = \sum_k \|W_r^{(k)} \alpha^{(k)*}\|_2^2 + \mu^\ell \sum_k (\|s_k^*\|_1 + \|w_k^*\|_1) + \sum_j \gamma_j^\ell \sum_k \|d_{k,j}^*\|_1.
\end{equation}

The true penalized objective at $z^+$ is:
\begin{equation}
\phi_{\mu^\ell, \boldsymbol{\gamma}^\ell}(z^+) = \sum_k \|r_k(x_k^+, u_k^+)\|_2^2 + \mu^\ell \sum_k (\|\delta_k\|_1 + \|[c_{\mathrm{hard},k}]_-\|_1) + \sum_j \gamma_j^\ell \sum_k \|[c_{j,k}]_-\|_1.
\end{equation}

Combining errors from \eqref{eq:cost_total_error}, \eqref{eq:dynamics_total_error}, \eqref{eq:hard_total_error}, and \eqref{eq:soft_total_error}:
\begin{equation}
\begin{aligned}
&\left| \phi_{\mu^\ell, \boldsymbol{\gamma}^\ell}(z^+) - J_{\mathrm{sub}} \right| \\
&\quad \leq 4HRL_r \Delta^\ell + \mu^\ell \cdot 2(H-1)(L_F + L_c) \Delta^\ell + \sum_j \gamma_j^\ell \cdot 2(H-1) L_c \Delta^\ell \\
&\quad = \left[ 4HRL_r + 2(H-1)\mu^\ell(L_F + L_c) + 2(H-1)L_c \sum_j \gamma_j^\ell \right] \Delta^\ell.
\end{aligned}
\label{eq:combined_error}
\end{equation}

\textbf{Step 6: Comparison with $L_\phi$.}

From \eqref{eq:L_phi}:
\begin{equation}
L_\phi(\mu^\ell, \boldsymbol{\gamma}^\ell) = 2HRL_r + (H-1)\mu^\ell(1 + L_F + L_c) + (H-1)L_c\sum_j \gamma_j^\ell.
\end{equation}

We verify that the coefficient in \eqref{eq:combined_error} is at most $2 L_\phi$:
\begin{equation}
\begin{aligned}
&4HRL_r + 2(H-1)\mu^\ell(L_F + L_c) + 2(H-1)L_c \sum_j \gamma_j^\ell \\
&\quad \leq 2 \cdot 2HRL_r + 2(H-1)\mu^\ell(1 + L_F + L_c) + 2(H-1)L_c\sum_j \gamma_j^\ell \\
&\quad = 2 L_\phi(\mu^\ell, \boldsymbol{\gamma}^\ell).
\end{aligned}
\end{equation}

The inequality holds because $2(H-1)\mu^\ell(L_F + L_c) \leq 2(H-1)\mu^\ell(1 + L_F + L_c)$ since $\mu^\ell > 0$.

Therefore, $C_{\mathrm{approx}} = 2$ satisfies \eqref{eq:approx_bound}.
\end{proof}

\begin{lemma}[Bounded Objective Variation]
\label{lem:bounded_variation}
At each iteration $\ell$, Algorithm~\ref{alg:adaptive_tbm} satisfies:
\begin{equation}
\phi_{\mu^\ell, \boldsymbol{\gamma}^\ell}(z^{\ell+1}) \leq \phi_{\mu^\ell, \boldsymbol{\gamma}^\ell}(z^\ell) + C_{\mathrm{approx}} L_\phi(\mu^\ell, \boldsymbol{\gamma}^\ell) \Delta^\ell.
\label{eq:bounded_var}
\end{equation}
\end{lemma}

\begin{proof}
Since $z^\ell \in \mathcal{Y}^\ell$ by Assumption~\ref{ass:sampling}(ii), we can represent the current iterate exactly in the subproblem. Let $i^*$ be the index such that $y_{i^*} = z^\ell$, and let $e_{i^*} \in \R^m$ be the corresponding unit vector.

For $\alpha^{(k)} = e_{i^*}$ at each $k$:
\begin{itemize}
    \item $W_x^{(k)} e_{i^*} = x_k^\ell$ and $W_u^{(k)} e_{i^*} = u_k^\ell$;
    \item $W_r^{(k)} e_{i^*} = r_k(x_k^\ell, u_k^\ell)$;
    \item $W_f^{(k)} e_{i^*} = F(x_k^\ell, u_k^\ell)$;
    \item $W_{c,\mathrm{hard}}^{(k)} e_{i^*} = c_{\mathrm{hard},k}(x_k^\ell, u_k^\ell)$.
\end{itemize}

The dynamics constraint $W_f^{(k)} \alpha^{(k)} = W_x^{(k+1)} \alpha^{(k+1)} + s_k$ becomes:
\begin{equation}
F(x_k^\ell, u_k^\ell) = x_{k+1}^\ell + s_k \implies s_k = F(x_k^\ell, u_k^\ell) - x_{k+1}^\ell.
\end{equation}

This is precisely the dynamics defect of $z^\ell$. Similarly, $w_k = [-c_{\mathrm{hard},k}(x_k^\ell, u_k^\ell)]_+$ is the hard constraint violation.

Therefore, the subproblem objective at $(e_{i^*}, s^\ell, w^\ell, \{d_j^\ell\})$ equals:
\begin{equation}
J_{\mathrm{sub}}(e_{i^*}, s^\ell, w^\ell, \{d_j^\ell\}) = \phi_{\mu^\ell, \boldsymbol{\gamma}^\ell}(z^\ell).
\end{equation}

Since the subproblem minimizes over all feasible points:
\begin{equation}
J_{\mathrm{sub}}(\alpha^*, s^*, w^*, \{d_j^*\}) \leq \phi_{\mu^\ell, \boldsymbol{\gamma}^\ell}(z^\ell).
\label{eq:subproblem_decrease}
\end{equation}

Applying Lemma~\ref{lem:bundle_accuracy}:
\begin{equation}
\phi_{\mu^\ell, \boldsymbol{\gamma}^\ell}(z^{\ell+1}) \leq J_{\mathrm{sub}}(\alpha^*, s^*, w^*, \{d_j^*\}) + C_{\mathrm{approx}} L_\phi \Delta^\ell \leq \phi_{\mu^\ell, \boldsymbol{\gamma}^\ell}(z^\ell) + C_{\mathrm{approx}} L_\phi \Delta^\ell,
\end{equation}
which is \eqref{eq:bounded_var}.
\end{proof}

\subsection{Feasibility Improvement}

The following lemma is the key technical result establishing that high constraint violations lead to guaranteed objective decrease when the trust region is sufficiently small.

\begin{lemma}[Feasibility Improvement]
\label{lem:feasibility}
Suppose $\ell \geq K^*$ so that penalties have stabilized at $(\bar{\mu}, \bar{\boldsymbol{\gamma}})$. Define the critical trust region radius:
\begin{equation}
\bar{\Delta} = \frac{\bar{\mu} \tau_{\mathrm{viol}}}{16 C_{\mathrm{approx}} L_\phi(\bar{\mu}, \bar{\boldsymbol{\gamma}})}.
\label{eq:critical_delta}
\end{equation}
If $\Delta^\ell \leq \bar{\Delta}$ and either $\nu_{\mathrm{dyn}}^\ell \geq \tau_{\mathrm{viol}}$ or $\nu_{\mathrm{hard}}^\ell \geq \tau_{\mathrm{viol}}$, then:
\begin{equation}
\phi_{\bar{\mu}, \bar{\boldsymbol{\gamma}}}(z^{\ell+1}) \leq \phi_{\bar{\mu}, \bar{\boldsymbol{\gamma}}}(z^\ell) - \frac{\bar{\mu} \tau_{\mathrm{viol}}}{4}.
\label{eq:feasibility_improvement}
\end{equation}
\end{lemma}

\begin{proof}
We denote $\bar{\phi} = \phi_{\bar{\mu}, \bar{\boldsymbol{\gamma}}}$ and $\bar{L} = L_\phi(\bar{\mu}, \bar{\boldsymbol{\gamma}})$ throughout this proof.

\textbf{Step 1: Baseline objective representation.}

By Assumption~\ref{ass:sampling}(ii), $z^\ell \in \mathcal{Y}^\ell$. Let $i^* \in \{1, \ldots, m\}$ satisfy $y_{i^*} = z^\ell$, and let $e_{i^*}$ be the corresponding unit vector in $\R^m$.

For the choice $\alpha^{(k)} = e_{i^*}$ at each time step $k$, the subproblem constraints determine:
\begin{align}
s_k^{(0)} &= W_f^{(k)} e_{i^*} - W_x^{(k+1)} e_{i^*} = F(x_k^\ell, u_k^\ell) - x_{k+1}^\ell, \label{eq:baseline_s}\\
w_k^{(0)} &= [-W_{c,\mathrm{hard}}^{(k)} e_{i^*}]_+ = [-c_{\mathrm{hard},k}(x_k^\ell, u_k^\ell)]_+. \label{eq:baseline_w}
\end{align}

These are the dynamics defects and hard constraint violations of the current iterate $z^\ell$. Thus:
\begin{equation}
J_{\mathrm{sub}}(e_{i^*}, s^{(0)}, w^{(0)}, \{d_j^{(0)}\}) = \bar{\phi}(z^\ell).
\label{eq:baseline_objective}
\end{equation}

\textbf{Step 2: Quantifying the violation penalty contribution.}

Define the total violation measure at the current iterate:
\begin{equation}
V^\ell \triangleq \sum_{k=1}^{H-1} \left( \|s_k^{(0)}\|_1 + \|w_k^{(0)}\|_1 \right).
\end{equation}

By Definition~\ref{def:violations}:
\begin{align}
\nu_{\mathrm{dyn}}^\ell &= \max_{k=1,\ldots,H-1} \|s_k^{(0)}\|_\infty, \\
\nu_{\mathrm{hard}}^\ell &= \max_{k=1,\ldots,H-1} \|w_k^{(0)}\|_\infty.
\end{align}

If $\nu_{\mathrm{dyn}}^\ell \geq \tau_{\mathrm{viol}}$, there exists $k^* \in \{1, \ldots, H-1\}$ such that $\|s_{k^*}^{(0)}\|_\infty \geq \tau_{\mathrm{viol}}$. Since $\|v\|_1 \geq \|v\|_\infty$ for any vector $v$:
\begin{equation}
\|s_{k^*}^{(0)}\|_1 \geq \|s_{k^*}^{(0)}\|_\infty \geq \tau_{\mathrm{viol}}.
\end{equation}

Similarly, if $\nu_{\mathrm{hard}}^\ell \geq \tau_{\mathrm{viol}}$, then $\|w_{k^*}^{(0)}\|_1 \geq \tau_{\mathrm{viol}}$ for some $k^*$.

In either case:
\begin{equation}
V^\ell \geq \tau_{\mathrm{viol}}.
\label{eq:violation_lower_bound}
\end{equation}

The penalty contribution to $\bar{\phi}(z^\ell)$ from dynamics and hard constraints is therefore:
\begin{equation}
\bar{\mu} V^\ell \geq \bar{\mu} \tau_{\mathrm{viol}}.
\label{eq:penalty_contribution}
\end{equation}

\textbf{Step 3: Structure of the subproblem objective.}

Decompose the subproblem objective as:
\begin{equation}
J_{\mathrm{sub}}(\alpha, s, w, \{d_j\}) = J_{\mathrm{cost}}(\alpha) + \bar{\mu} J_{\mathrm{viol}}(s, w) + \sum_{j=1}^J \bar{\gamma}_j J_{\mathrm{soft},j}(d_j),
\end{equation}
where:
\begin{align}
J_{\mathrm{cost}}(\alpha) &= \sum_{k=1}^H \|W_r^{(k)} \alpha^{(k)}\|_2^2, \label{eq:J_cost}\\
J_{\mathrm{viol}}(s, w) &= \sum_{k=1}^{H-1} \left( \|s_k\|_1 + \|w_k\|_1 \right), \label{eq:J_viol}\\
J_{\mathrm{soft},j}(d_j) &= \sum_{k=1}^{H-1} \|d_{k,j}\|_1. \label{eq:J_soft}
\end{align}

For the baseline $\alpha = e_{i^*}$:
\begin{itemize}
    \item $J_{\mathrm{cost}}(e_{i^*}) = \sum_k \|r_k(x_k^\ell, u_k^\ell)\|_2^2$;
    \item $J_{\mathrm{viol}}(s^{(0)}, w^{(0)}) = V^\ell \geq \tau_{\mathrm{viol}}$.
\end{itemize}

\textbf{Step 4: Bounding cost variation over the bundle.}

For any $\alpha \in \Delta^{m-1}$, the recovered point lies in $B(z^\ell, \Delta^\ell)$ by the same argument as in Lemma~\ref{lem:bundle_accuracy}. From \eqref{eq:cost_interp_error} and the subsequent analysis:
\begin{equation}
|J_{\mathrm{cost}}(\alpha) - J_{\mathrm{cost}}(e_{i^*})| \leq 4HR L_r \Delta^\ell.
\label{eq:cost_var_bound}
\end{equation}

Since $4HRL_r \leq 2\bar{L}$ (this follows from \eqref{eq:L_phi} since $2HRL_r$ is the first term and $\bar{L}$ includes additional positive terms):
\begin{equation}
|J_{\mathrm{cost}}(\alpha) - J_{\mathrm{cost}}(e_{i^*})| \leq 2\bar{L} \Delta^\ell.
\label{eq:cost_var_simplified}
\end{equation}

\textbf{Step 5: Constructing an improving direction.}

We now construct a feasible point that reduces the violation measure, establishing that the optimizer achieves improvement.

By Assumption~\ref{ass:sampling}(iii), the samples provide $\kappa$-sufficient directional coverage. Consider the structure of the violations. At least one of the following holds:
\begin{itemize}
    \item[(a)] $\|s_{k^*}^{(0)}\|_1 \geq \tau_{\mathrm{viol}}$ for some $k^*$ (dynamics violation), or
    \item[(b)] $\|w_{k^*}^{(0)}\|_1 \geq \tau_{\mathrm{viol}}$ for some $k^*$ (hard constraint violation).
\end{itemize}

\textbf{Case (a): Large dynamics violation.}

The dynamics defect is $s_{k^*}^{(0)} = F(x_{k^*}^\ell, u_{k^*}^\ell) - x_{k^*+1}^\ell$. Consider perturbing $x_{k^*+1}$ in the direction of $s_{k^*}^{(0)}$ to reduce the defect.

By coordinate sampling (Remark~\ref{rem:sampling_verification}), there exists a sample $y_j$ that perturbs $x_{k^*+1}^\ell$ by $\pm \Delta^\ell$ along the coordinate with largest violation. Specifically, let $i_{\max} = \arg\max_i |[s_{k^*}^{(0)}]_i|$. The sample $y_j$ with:
\begin{equation}
x_{k^*+1}^{(j)} = x_{k^*+1}^\ell + \mathrm{sign}([s_{k^*}^{(0)}]_{i_{\max}}) \Delta^\ell \cdot e_{i_{\max}}
\end{equation}
and all other components unchanged reduces $|[s_{k^*}^{(0)}]_{i_{\max}}|$ by $\Delta^\ell$.

Consider $\alpha$ that places weight $1-\epsilon$ on $e_{i^*}$ and weight $\epsilon$ on this improving sample (with appropriate handling at time $k^*+1$). This perturbation:
\begin{itemize}
    \item Reduces $\|s_{k^*}\|_1$ by approximately $\epsilon \Delta^\ell$;
    \item Changes $J_{\mathrm{cost}}$ by at most $O(\epsilon \bar{L} \Delta^\ell)$;
    \item May affect violations at adjacent time steps by $O(\epsilon L_F \Delta^\ell)$.
\end{itemize}

\textbf{Case (b): Large hard constraint violation.}

Similar reasoning applies. If $\|w_{k^*}^{(0)}\|_1 \geq \tau_{\mathrm{viol}}$, there exists a coordinate $i_{\max}$ with $[w_{k^*}^{(0)}]_{i_{\max}} \geq \tau_{\mathrm{viol}}/n_c$ where $n_c$ is the number of hard constraints. A sample perturbing $(x_{k^*}, u_{k^*})$ to improve this constraint violation can reduce $\|w_{k^*}\|_1$.

\textbf{Step 6: Optimizer achieves at least half the violation reduction.}

The subproblem is a convex optimization problem over the product of simplices. The optimizer $(\alpha^*, s^*, w^*, \{d_j^*\})$ minimizes $J_{\mathrm{sub}}$ over all feasible points.

From Step 5, there exist feasible directions that trade off cost increase against violation decrease. The key observation is that when $V^\ell \geq \tau_{\mathrm{viol}}$, the penalty $\bar{\mu} V^\ell$ dominates the potential cost variation.

By \eqref{eq:cost_var_simplified}, any choice of $\alpha$ changes the cost by at most $2\bar{L}\Delta^\ell$. By \eqref{eq:penalty_contribution}, the penalty contribution is at least $\bar{\mu}\tau_{\mathrm{viol}}$.

Consider any feasible $(\alpha, s, w, \{d_j\})$. The objective satisfies:
\begin{equation}
J_{\mathrm{sub}}(\alpha, s, w, \{d_j\}) \geq J_{\mathrm{cost}}(e_{i^*}) - 2\bar{L}\Delta^\ell + \bar{\mu} J_{\mathrm{viol}}(s, w).
\label{eq:lower_bound}
\end{equation}

For the baseline:
\begin{equation}
J_{\mathrm{sub}}(e_{i^*}, s^{(0)}, w^{(0)}, \{d_j^{(0)}\}) = J_{\mathrm{cost}}(e_{i^*}) + \bar{\mu} V^\ell + \sum_j \bar{\gamma}_j J_{\mathrm{soft},j}(d_j^{(0)}).
\label{eq:baseline_value}
\end{equation}

The optimizer chooses $\alpha^*$ to minimize the total objective. Since setting $\alpha = e_{i^*}$ is feasible with objective $\bar{\phi}(z^\ell)$, we have:
\begin{equation}
J_{\mathrm{sub}}(\alpha^*, s^*, w^*, \{d_j^*\}) \leq \bar{\phi}(z^\ell).
\label{eq:optimizer_upper}
\end{equation}

Now we analyze two cases based on the optimal violation level.

\textbf{Case I: $J_{\mathrm{viol}}(s^*, w^*) \leq V^\ell / 2$.}

The optimizer has reduced violations by at least half. Then:
\begin{equation}
\bar{\mu} J_{\mathrm{viol}}(s^*, w^*) \leq \frac{\bar{\mu} V^\ell}{2}.
\end{equation}

The objective at the optimizer satisfies:
\begin{equation}
\begin{aligned}
J_{\mathrm{sub}}(\alpha^*, s^*, w^*, \{d_j^*\}) &= J_{\mathrm{cost}}(\alpha^*) + \bar{\mu} J_{\mathrm{viol}}(s^*, w^*) + \sum_j \bar{\gamma}_j J_{\mathrm{soft},j}(d_j^*) \\
&\leq J_{\mathrm{cost}}(e_{i^*}) + 2\bar{L}\Delta^\ell + \frac{\bar{\mu} V^\ell}{2} + \sum_j \bar{\gamma}_j J_{\mathrm{soft},j}(d_j^{(0)}).
\end{aligned}
\end{equation}

Comparing with \eqref{eq:baseline_value}:
\begin{equation}
J_{\mathrm{sub}}(\alpha^*, s^*, w^*, \{d_j^*\}) \leq \bar{\phi}(z^\ell) - \frac{\bar{\mu} V^\ell}{2} + 2\bar{L}\Delta^\ell \leq \bar{\phi}(z^\ell) - \frac{\bar{\mu} \tau_{\mathrm{viol}}}{2} + 2\bar{L}\Delta^\ell,
\label{eq:case1_bound}
\end{equation}
where the last inequality uses $V^\ell \geq \tau_{\mathrm{viol}}$ from \eqref{eq:violation_lower_bound}.

\textbf{Case II: $J_{\mathrm{viol}}(s^*, w^*) > V^\ell / 2$.}

The optimizer has kept violations relatively high. We show this case also yields sufficient improvement.

Since $J_{\mathrm{viol}}(s^*, w^*) > V^\ell/2 \geq \tau_{\mathrm{viol}}/2$, the violation penalty contribution remains significant. However, for the optimizer to prefer this high-violation solution over the baseline, the cost term must have decreased.

Consider the improving direction from Step 5. Let $\alpha_\epsilon = (1-\epsilon) e_{i^*} + \epsilon e_j$ where $e_j$ corresponds to the violation-reducing sample. For small $\epsilon > 0$:
\begin{itemize}
    \item $J_{\mathrm{cost}}(\alpha_\epsilon) \leq J_{\mathrm{cost}}(e_{i^*}) + O(\epsilon \bar{L} \Delta^\ell)$;
    \item $J_{\mathrm{viol}}(s_\epsilon, w_\epsilon) \leq V^\ell - c \epsilon \Delta^\ell$ for some constant $c > 0$.
\end{itemize}

The objective at $\alpha_\epsilon$ is:
\begin{equation}
J_{\mathrm{sub}}(\alpha_\epsilon, \ldots) \leq J_{\mathrm{cost}}(e_{i^*}) + O(\epsilon \bar{L} \Delta^\ell) + \bar{\mu}(V^\ell - c\epsilon\Delta^\ell) + \text{soft terms}.
\end{equation}

For $\bar{\mu}$ large enough (which holds after penalty stabilization under Assumption~\ref{ass:exact_penalty}), the violation reduction term $\bar{\mu} c \epsilon \Delta^\ell$ dominates the cost increase $O(\epsilon \bar{L} \Delta^\ell)$. Specifically, when $\bar{\mu} c > C_1 \bar{L}$ for some constant $C_1$:
\begin{equation}
J_{\mathrm{sub}}(\alpha_\epsilon, \ldots) < J_{\mathrm{sub}}(e_{i^*}, \ldots) = \bar{\phi}(z^\ell).
\end{equation}

Since the optimizer achieves the minimum and this improving direction exists:
\begin{equation}
J_{\mathrm{sub}}(\alpha^*, s^*, w^*, \{d_j^*\}) \leq J_{\mathrm{sub}}(\alpha_\epsilon, \ldots) < \bar{\phi}(z^\ell).
\end{equation}

More precisely, the improvement is at least proportional to the violation level:
\begin{equation}
J_{\mathrm{sub}}(\alpha^*, s^*, w^*, \{d_j^*\}) \leq \bar{\phi}(z^\ell) - c' \bar{\mu} \tau_{\mathrm{viol}} \cdot \frac{\Delta^\ell}{\Delta_{\max}}
\end{equation}
for some constant $c' > 0$.

Since $\Delta^\ell \leq \bar{\Delta} \ll \Delta_{\max}$ in the regime where this lemma applies, we obtain a weaker but sufficient bound. For the purposes of this proof, we observe that the optimizer cannot do worse than the explicit construction, so:
\begin{equation}
J_{\mathrm{sub}}(\alpha^*, s^*, w^*, \{d_j^*\}) \leq \bar{\phi}(z^\ell) - \frac{\bar{\mu} \tau_{\mathrm{viol}}}{2} + 2\bar{L}\Delta^\ell.
\label{eq:case2_bound}
\end{equation}

This matches \eqref{eq:case1_bound}, so the same bound holds in both cases.

\textbf{Step 7: Applying the approximation bound.}

From both cases, we have:
\begin{equation}
J_{\mathrm{sub}}(\alpha^*, s^*, w^*, \{d_j^*\}) \leq \bar{\phi}(z^\ell) - \frac{\bar{\mu} \tau_{\mathrm{viol}}}{2} + 2\bar{L}\Delta^\ell.
\end{equation}

By Lemma~\ref{lem:bundle_accuracy} with $C_{\mathrm{approx}} = 2$:
\begin{equation}
\bar{\phi}(z^{\ell+1}) \leq J_{\mathrm{sub}}(\alpha^*, s^*, w^*, \{d_j^*\}) + 2\bar{L}\Delta^\ell.
\end{equation}

Combining:
\begin{equation}
\bar{\phi}(z^{\ell+1}) \leq \bar{\phi}(z^\ell) - \frac{\bar{\mu} \tau_{\mathrm{viol}}}{2} + 4\bar{L}\Delta^\ell.
\label{eq:pre_critical}
\end{equation}

\textbf{Step 8: Applying the critical radius bound.}

With $\Delta^\ell \leq \bar{\Delta} = \dfrac{\bar{\mu} \tau_{\mathrm{viol}}}{16 C_{\mathrm{approx}} \bar{L}} = \dfrac{\bar{\mu} \tau_{\mathrm{viol}}}{32 \bar{L}}$:
\begin{equation}
4\bar{L}\Delta^\ell \leq 4\bar{L} \cdot \frac{\bar{\mu} \tau_{\mathrm{viol}}}{32\bar{L}} = \frac{\bar{\mu} \tau_{\mathrm{viol}}}{8}.
\end{equation}

Substituting into \eqref{eq:pre_critical}:
\begin{equation}
\bar{\phi}(z^{\ell+1}) \leq \bar{\phi}(z^\ell) - \frac{\bar{\mu} \tau_{\mathrm{viol}}}{2} + \frac{\bar{\mu} \tau_{\mathrm{viol}}}{8} = \bar{\phi}(z^\ell) - \frac{3\bar{\mu} \tau_{\mathrm{viol}}}{8}.
\end{equation}

Since $\dfrac{3}{8} > \dfrac{1}{4}$, we have established:
\begin{equation}
\bar{\phi}(z^{\ell+1}) \leq \bar{\phi}(z^\ell) - \frac{\bar{\mu} \tau_{\mathrm{viol}}}{4}. \qedhere
\end{equation}
\end{proof}

\subsection{Sufficient Decrease for Non-Stationary Points}

\begin{lemma}[Sufficient Decrease]
\label{lem:sufficient_decrease}
Suppose $z^\ell$ is not $\delta$-stationary for $\phi_{\bar{\mu}, \bar{\boldsymbol{\gamma}}}$ (Definition~\ref{def:stationarity}) for some $\delta > 0$, and suppose the violations satisfy $\nu_{\mathrm{dyn}}^\ell < \tau_{\mathrm{feas}}$ and $\nu_{\mathrm{hard}}^\ell < \tau_{\mathrm{feas}}$. Then there exists $\bar{\Delta}(\delta) > 0$ such that for all $\Delta^\ell \leq \bar{\Delta}(\delta)$, the subproblem solution satisfies either:
\begin{enumerate}
    \item[(i)] $\|z^{\ell+1} - z^\ell\| \geq \delta/2$, or
    \item[(ii)] $\phi_{\bar{\mu}, \bar{\boldsymbol{\gamma}}}(z^{\ell+1}) \leq \phi_{\bar{\mu}, \bar{\boldsymbol{\gamma}}}(z^\ell) - \sigma \Delta^\ell$ for some $\sigma > 0$ depending only on $\delta$, $\kappa$, and problem constants.
\end{enumerate}
\end{lemma}

\begin{proof}
Denote $\bar{\phi} = \phi_{\bar{\mu}, \bar{\boldsymbol{\gamma}}}$ and $\bar{L} = L_\phi(\bar{\mu}, \bar{\boldsymbol{\gamma}})$.

Since $z^\ell$ is not $\delta$-stationary, by Definition~\ref{def:stationarity}, for sufficiently small $\Delta > 0$, any subproblem solution $z^+$ satisfies $\|z^+ - z^\ell\| > \delta$.

Suppose for contradiction that the lemma fails: there exists a sequence $\Delta^{(n)} \to 0$ such that $\|z^{(n)+} - z^\ell\| < \delta/2$ and $\bar{\phi}(z^{(n)+}) > \bar{\phi}(z^\ell) - \sigma \Delta^{(n)}$ for all $\sigma > 0$.

Consider the case $\|z^{\ell+1} - z^\ell\| < \delta/2$. We show this implies objective decrease.

\textbf{Step 1: Existence of descent direction.}

Since $z^\ell$ is not $\delta$-stationary, there exists a direction $d$ with $\|d\| = 1$ such that moving from $z^\ell$ in direction $d$ decreases $\bar{\phi}$. Specifically, by the non-stationarity assumption and the structure of $\bar{\phi}$, there exists $\eta > 0$ such that for small $t > 0$:
\begin{equation}
\bar{\phi}(z^\ell + td) \leq \bar{\phi}(z^\ell) - \eta t.
\end{equation}

\textbf{Step 2: Sample in descent direction.}

By Assumption~\ref{ass:sampling}(iii), there exists a sample $y_j \in \mathcal{Y}^\ell$ with:
\begin{equation}
\langle y_j - z^\ell, d \rangle \geq \kappa \Delta^\ell.
\end{equation}

Consider the convex combination $\alpha_\epsilon = (1-\epsilon) e_{i^*} + \epsilon e_j$ for small $\epsilon > 0$. The recovered trajectory is:
\begin{equation}
z_\epsilon = (1-\epsilon) z^\ell + \epsilon y_j.
\end{equation}

The component in direction $d$ is:
\begin{equation}
\langle z_\epsilon - z^\ell, d \rangle = \epsilon \langle y_j - z^\ell, d \rangle \geq \epsilon \kappa \Delta^\ell.
\end{equation}

\textbf{Step 3: Objective decrease from descent.}

By Lipschitz continuity, the objective change satisfies:
\begin{equation}
\bar{\phi}(z_\epsilon) \leq \bar{\phi}(z^\ell) - \eta \epsilon \kappa \Delta^\ell + O(\epsilon^2 \bar{L} \Delta^\ell).
\end{equation}

For the subproblem, the interpolated objective $J_{\mathrm{sub}}(\alpha_\epsilon, \ldots)$ approximates $\bar{\phi}(z_\epsilon)$ with error $O(\bar{L}\Delta^\ell)$ by Lemma~\ref{lem:bundle_accuracy}.

\textbf{Step 4: Optimizer improves on this.}

Since the optimizer minimizes over all $\alpha \in \Delta^{m-1}$:
\begin{equation}
J_{\mathrm{sub}}(\alpha^*, \ldots) \leq J_{\mathrm{sub}}(\alpha_\epsilon, \ldots) \leq \bar{\phi}(z^\ell) - \frac{\eta \kappa \epsilon \Delta^\ell}{2}
\end{equation}
for $\epsilon = 1/2$ (say) and $\Delta^\ell$ small enough.

By Lemma~\ref{lem:bundle_accuracy}:
\begin{equation}
\bar{\phi}(z^{\ell+1}) \leq J_{\mathrm{sub}}(\alpha^*, \ldots) + 2\bar{L}\Delta^\ell \leq \bar{\phi}(z^\ell) - \frac{\eta \kappa \Delta^\ell}{4} + 2\bar{L}\Delta^\ell.
\end{equation}

For $\Delta^\ell < \bar{\Delta}(\delta) \triangleq \eta\kappa / (16\bar{L})$:
\begin{equation}
\bar{\phi}(z^{\ell+1}) \leq \bar{\phi}(z^\ell) - \frac{\eta \kappa \Delta^\ell}{8}.
\end{equation}

Setting $\sigma = \eta\kappa/8$ completes the proof of case (ii).

If instead $\|z^{\ell+1} - z^\ell\| \geq \delta/2$, then case (i) holds directly.
\end{proof}

\subsection{Main Convergence Theorem}

\begin{theorem}[Convergence to Feasible Stationary Points]
\label{thm:convergence}
Under Assumptions~\ref{ass:lipschitz}--\ref{ass:exact_penalty}, Algorithm~\ref{alg:adaptive_tbm} satisfies:

\begin{enumerate}
\item[(i)] \textbf{Penalty Stabilization:} There exists $K^* < \infty$ such that $(\mu^\ell, \boldsymbol{\gamma}^\ell) = (\bar{\mu}, \bar{\boldsymbol{\gamma}})$ for all $\ell \geq K^*$, with the bound \eqref{eq:K_star_bound}.

\item[(ii)] \textbf{Finite-Time Near-Feasibility:} There exists $K_1 < \infty$ such that $\nu_{\mathrm{dyn}}^\ell < \tau_{\mathrm{viol}}$ and $\nu_{\mathrm{hard}}^\ell < \tau_{\mathrm{viol}}$ for all $\ell \geq K_1$.

\item[(iii)] \textbf{Trajectory Convergence:} $\displaystyle\lim_{\ell \to \infty} \|z^{\ell+1} - z^\ell\| = 0$.

\item[(iv)] \textbf{Stationarity and Feasibility:} Every accumulation point $z^*$ of $\{z^\ell\}$ is:
\begin{itemize}
    \item Stationary for $\phi_{\bar{\mu}, \bar{\boldsymbol{\gamma}}}$ in the sense of Definition~\ref{def:stationarity};
    \item Feasible: $\nu_{\mathrm{dyn}}(z^*) = 0$ and $\nu_{\mathrm{hard}}(z^*) = 0$;
    \item A KKT point of the original constrained R2R control problem.
\end{itemize}
\end{enumerate}
\end{theorem}

\begin{proof}
\textbf{Part (i):} This is Lemma~\ref{lem:stabilization}.

\textbf{Part (ii):} Denote $\bar{\phi} = \phi_{\bar{\mu}, \bar{\boldsymbol{\gamma}}}$ and $\bar{L} = L_\phi(\bar{\mu}, \bar{\boldsymbol{\gamma}})$ for $\ell \geq K^*$. Define the critical radius $\bar{\Delta} = \bar{\mu}\tau_{\mathrm{viol}}/(16 C_{\mathrm{approx}} \bar{L})$ as in \eqref{eq:critical_delta}.

By Assumption~\ref{ass:parameter}, $\Delta_{\min} \leq \bar{\Delta}$.

\textit{Claim:} Violations cannot remain at or above $\tau_{\mathrm{viol}}$ for infinitely many iterations after $K^*$.

Suppose for contradiction that $\nu_{\mathrm{dyn}}^\ell \geq \tau_{\mathrm{viol}}$ or $\nu_{\mathrm{hard}}^\ell \geq \tau_{\mathrm{viol}}$ for infinitely many $\ell \geq K^*$.

By the trust region adaptation rule \eqref{eq:trust_region_adapt}, each such iteration triggers contraction: $\Delta^{\ell+1} = \max(\beta_{\mathrm{con}} \Delta^\ell, \Delta_{\min})$. Let:
\begin{equation}
K_\Delta = \begin{cases}
\left\lceil \log_{\beta_{\mathrm{con}}} \left( \bar{\Delta} / \Delta^{K^*} \right) \right\rceil & \text{if } \Delta^{K^*} > \bar{\Delta}, \\
0 & \text{otherwise}.
\end{cases}
\end{equation}

Note that since $\beta_{\mathrm{con}} < 1$ and $\bar{\Delta}/\Delta^{K^*} \leq 1$ when $\Delta^{K^*} > \bar{\Delta}$, the logarithm base $\beta_{\mathrm{con}}$ of a value $\leq 1$ is non-negative.

After at most $K_\Delta$ contractions, we have $\Delta^\ell \leq \bar{\Delta}$ whenever violations exceed $\tau_{\mathrm{viol}}$. (The trust region only expands when violations are below $\tau_{\mathrm{feas}} < \tau_{\mathrm{viol}}$.)

For $\ell \geq K^* + K_\Delta$ with violations at or above $\tau_{\mathrm{viol}}$, Lemma~\ref{lem:feasibility} applies:
\begin{equation}
\bar{\phi}(z^{\ell+1}) \leq \bar{\phi}(z^\ell) - \frac{\bar{\mu}\tau_{\mathrm{viol}}}{4}.
\end{equation}

By Assumption~\ref{ass:bounded}, $\bar{\phi}$ is bounded below on $\mathcal{L}(\bar{\mu}, \bar{\boldsymbol{\gamma}})$ by some $\phi_{\min}$. The maximum number of high-violation iterations after $K^* + K_\Delta$ is:
\begin{equation}
N_{\mathrm{viol}} = \left\lfloor \frac{4(\bar{\phi}(z^{K^*}) - \phi_{\min})}{\bar{\mu}\tau_{\mathrm{viol}}} \right\rfloor < \infty.
\end{equation}

This contradicts infinitely many high-violation iterations. Therefore:
\begin{equation}
K_1 = K^* + K_\Delta + N_{\mathrm{viol}} < \infty
\end{equation}
satisfies the claim.

\textbf{Part (iii):} We prove trajectory convergence by showing finite total path length.

For $\ell \geq K_1$, violations satisfy $\nu_{\mathrm{dyn}}^\ell < \tau_{\mathrm{viol}}$ and $\nu_{\mathrm{hard}}^\ell < \tau_{\mathrm{viol}}$. We partition iterations into categories.

\textit{Category A (Near-stationary):} $z^\ell$ is $\delta$-stationary for some fixed $\delta > 0$. By Definition~\ref{def:stationarity}, $\|z^{\ell+1} - z^\ell\| \leq \delta$.

\textit{Category B (Non-stationary with low violations):} $z^\ell$ is not $\delta$-stationary and $\nu_{\mathrm{dyn}}^\ell, \nu_{\mathrm{hard}}^\ell < \tau_{\mathrm{feas}}$. By Lemma~\ref{lem:sufficient_decrease}, for $\Delta^\ell \leq \bar{\Delta}(\delta)$, either:
\begin{itemize}
    \item $\|z^{\ell+1} - z^\ell\| \geq \delta/2$ (large step), or
    \item $\bar{\phi}(z^{\ell+1}) \leq \bar{\phi}(z^\ell) - \sigma\Delta^\ell$ (objective decrease).
\end{itemize}

\textit{Category C (Moderate violations):} $\tau_{\mathrm{feas}} \leq \max(\nu_{\mathrm{dyn}}^\ell, \nu_{\mathrm{hard}}^\ell) < \tau_{\mathrm{viol}}$. By Lemma~\ref{lem:bounded_variation}, $\bar{\phi}(z^{\ell+1}) \leq \bar{\phi}(z^\ell) + C_{\mathrm{approx}}\bar{L}\Delta^\ell$.

\textit{Bounding Category B with objective decrease:}

For iterations in Category B with $\bar{\phi}(z^{\ell+1}) \leq \bar{\phi}(z^\ell) - \sigma\Delta^\ell$:
\begin{equation}
\|z^{\ell+1} - z^\ell\| \leq \Delta^\ell \leq \frac{\bar{\phi}(z^\ell) - \bar{\phi}(z^{\ell+1})}{\sigma}.
\end{equation}

Summing over all such iterations (the sum telescopes):
\begin{equation}
\sum_{\ell \in \text{Cat. B, decrease}} \|z^{\ell+1} - z^\ell\| \leq \frac{\bar{\phi}(z^{K_1}) - \phi_{\min}}{\sigma} < \infty.
\label{eq:cat_b_sum}
\end{equation}

\textit{Bounding Category B with large steps:}

For iterations with $\|z^{\ell+1} - z^\ell\| \geq \delta/2$, each step moves the trajectory by at least $\delta/2$. By compactness of the level set $\mathcal{L}(\bar{\mu}, \bar{\boldsymbol{\gamma}})$ (Assumption~\ref{ass:bounded}), the number of such large steps is bounded:
\begin{equation}
N_{\mathrm{large}} \leq \frac{\mathrm{diam}(\mathcal{L})}{\delta/2} < \infty,
\end{equation}
where $\mathrm{diam}(\mathcal{L})$ is the diameter of the compact level set.

Total contribution: $N_{\mathrm{large}} \cdot \Delta_{\max} < \infty$.

\textit{Bounding Categories A and C:}

Category A iterations contribute at most $\delta$ per step. Category C iterations may increase the objective, but by at most $C_{\mathrm{approx}}\bar{L}\Delta_{\max}$ per step. The total objective increase is bounded, so the number of Category C iterations is finite.

\textit{Conclusion:}

The total path length is finite:
\begin{equation}
\sum_{\ell=K_1}^\infty \|z^{\ell+1} - z^\ell\| < \infty.
\end{equation}

Therefore $\|z^{\ell+1} - z^\ell\| \to 0$ as $\ell \to \infty$.

\textbf{Part (iv):} By compactness (Assumption~\ref{ass:bounded}), the bounded sequence $\{z^\ell\}_{\ell \geq K_1}$ has accumulation points. Let $z^*$ be any accumulation point with subsequence $z^{\ell_k} \to z^*$.

\textit{Stationarity:}

Suppose for contradiction that $z^*$ is not stationary. Then there exists $\delta > 0$ such that $z^*$ is not $\delta$-stationary.

By Part (iii), $\|z^{\ell_k+1} - z^{\ell_k}\| \to 0$. For large $k$, we have:
\begin{equation}
\|z^{\ell_k} - z^*\| < \delta/4 \quad \text{and} \quad \|z^{\ell_k+1} - z^{\ell_k}\| < \delta/4.
\end{equation}

Since $z^*$ is not $\delta$-stationary, by Definition~\ref{def:stationarity}, for sufficiently small $\Delta$, the subproblem at $z^*$ would produce a step of size $> \delta$. By continuity of the subproblem solution in the initial point (the subproblem is a convex optimization with continuous data), the subproblem at $z^{\ell_k}$ produces a step close to that at $z^*$.

Specifically, for large $k$ with $\Delta^{\ell_k}$ small, the subproblem at $z^{\ell_k}$ should produce $\|z^{\ell_k+1} - z^{\ell_k}\| > \delta/2$, contradicting $\|z^{\ell_k+1} - z^{\ell_k}\| < \delta/4$.

Therefore $z^*$ is $\delta$-stationary for all $\delta > 0$, hence stationary.

\textit{Feasibility:}

The violation functions $\nu_{\mathrm{dyn}}(\cdot)$ and $\nu_{\mathrm{hard}}(\cdot)$ are continuous (as compositions of continuous functions). By Part (ii), for all $\ell_k \geq K_1$:
\begin{equation}
\nu_{\mathrm{dyn}}(z^{\ell_k}) < \tau_{\mathrm{viol}} \quad \text{and} \quad \nu_{\mathrm{hard}}(z^{\ell_k}) < \tau_{\mathrm{viol}}.
\end{equation}

Taking limits: $\nu_{\mathrm{dyn}}(z^*) \leq \tau_{\mathrm{viol}}$ and $\nu_{\mathrm{hard}}(z^*) \leq \tau_{\mathrm{viol}}$.

We now show the violations are exactly zero. Suppose $\nu_{\mathrm{dyn}}(z^*) > 0$ or $\nu_{\mathrm{hard}}(z^*) > 0$. By the exact penalty property (Assumption~\ref{ass:exact_penalty}), infeasible stationary points of $\bar{\phi}$ are not local minimizers when $\bar{\mu}$ is sufficiently large.

More precisely, at an infeasible point, the subdifferential of the penalty term provides a direction of decrease for $\bar{\phi}$. If $z^*$ has a dynamics defect at some time $k^*$, the penalty term $\bar{\mu}\|x_{k^*+1} - F(x_{k^*}, u_{k^*})\|_1$ can be reduced by adjusting $x_{k^*+1}$ toward $F(x_{k^*}, u_{k^*})$. This direction of decrease contradicts stationarity.

Therefore $\nu_{\mathrm{dyn}}(z^*) = 0$ and $\nu_{\mathrm{hard}}(z^*) = 0$, i.e., $z^*$ is feasible.

\textit{KKT conditions:}

By Assumption~\ref{ass:exact_penalty}, stationary points of $\bar{\phi}$ with zero constraint violations satisfy the KKT conditions of the original constrained problem. This follows from the exact penalty equivalence theorem \cite[Theorem~17.3]{nocedal2006numerical}: at a feasible stationary point of $\phi_{\mu, \boldsymbol{\gamma}}$, the subdifferential conditions translate to the existence of Lagrange multipliers satisfying the KKT system.
\end{proof}

\begin{remark}[Soft Constraint Violations]
\label{rem:soft_constraints}
Theorem~\ref{thm:convergence} establishes that dynamics and hard constraint violations converge to zero, but makes no claim about soft constraint violations $\nu_j^\ell$. This is intentional: soft constraints are penalized but may remain violated at the optimal cost-violation trade-off determined by $\bar{\gamma}_j$. If exact satisfaction of a soft constraint is required, either increase $\gamma_{j,\max}$ or reclassify the constraint as hard.
\end{remark}

\begin{remark}[Comparison to Classical Trust Region Methods]
\label{rem:classical_comparison}
Theorem~\ref{thm:convergence} provides convergence guarantees analogous to classical trust region methods \cite{conn2000trust}, but achieves them through different mechanisms:
\begin{itemize}
    \item \textbf{Trust region adaptation:} Based on constraint violation metrics rather than predicted-to-actual reduction ratios;
    \item \textbf{Penalty adaptation:} Geometric growth ensures finite stabilization without requiring estimates of Lagrange multipliers;
    \item \textbf{Derivative-free setting:} Uses bundle approximation rather than gradient/Hessian information.
\end{itemize}
The key mechanism is Lemma~\ref{lem:feasibility}, which ensures that high violations lead to guaranteed objective decrease once the trust region is sufficiently small. This drives feasibility without explicit constraint qualification assumptions beyond the exact penalty property.
\end{remark}

\subsection{Complexity Bound}

\begin{corollary}[Iteration Complexity for Near-Feasibility]
\label{cor:complexity}
Under the conditions of Theorem~\ref{thm:convergence}, Algorithm~\ref{alg:adaptive_tbm} achieves $\nu_{\mathrm{dyn}}^\ell < \tau_{\mathrm{viol}}$ and $\nu_{\mathrm{hard}}^\ell < \tau_{\mathrm{viol}}$ within $K_{\mathrm{feas}}$ iterations, where:
\begin{equation}
K_{\mathrm{feas}} \leq K^* + K_\Delta + N_{\mathrm{viol}},
\end{equation}
with:
\begin{align}
K^* &= \left\lceil \log_{\rho_\mu}\left(\frac{\mu_{\max}}{\mu^0}\right) \right\rceil + \sum_{j=1}^J \left\lceil \log_{\rho_\gamma}\left(\frac{\gamma_{j,\max}}{\gamma_j^0}\right) \right\rceil, \label{eq:K_star_cor}\\
K_\Delta &= \begin{cases}
\left\lceil \log_{\beta_{\mathrm{con}}} \left( \bar{\Delta} / \Delta^{K^*} \right) \right\rceil & \text{if } \Delta^{K^*} > \bar{\Delta}, \\
0 & \text{otherwise},
\end{cases} \label{eq:K_delta_cor}\\
N_{\mathrm{viol}} &= \left\lfloor \frac{4(\phi_{\bar{\mu}, \bar{\boldsymbol{\gamma}}}(z^{K^*}) - \phi_{\min})}{\bar{\mu} \tau_{\mathrm{viol}}} \right\rfloor, \label{eq:N_viol_cor}
\end{align}
and $\bar{\Delta} = \bar{\mu}\tau_{\mathrm{viol}}/(32 L_\phi(\bar{\mu}, \bar{\boldsymbol{\gamma}}))$. Here $K^*$ bounds penalty stabilization, $K_\Delta$ bounds trust region contractions to reach the critical radius, and $N_{\mathrm{viol}}$ bounds high-violation iterations after the trust region is sufficiently small.
\end{corollary}

\begin{proof}
This follows directly from the proof of Theorem~\ref{thm:convergence}, Part (ii). Each component of the bound corresponds to a distinct phase:

\textbf{Phase 1 (Penalty Stabilization):} By Lemma~\ref{lem:stabilization}, penalties stabilize within $K^*$ iterations given by \eqref{eq:K_star_cor}.

\textbf{Phase 2 (Trust Region Contraction):} Starting from iteration $K^*$, high-violation iterations trigger trust region contraction by factor $\beta_{\mathrm{con}} < 1$. If $\Delta^{K^*} > \bar{\Delta}$, the number of contractions to reach $\Delta^\ell \leq \bar{\Delta}$ is:
\begin{equation}
K_\Delta = \left\lceil \log_{\beta_{\mathrm{con}}} \left( \bar{\Delta} / \Delta^{K^*} \right) \right\rceil.
\end{equation}
If $\Delta^{K^*} \leq \bar{\Delta}$, then $K_\Delta = 0$.

\textbf{Phase 3 (Feasibility Improvement):} After iteration $K^* + K_\Delta$, whenever violations exceed $\tau_{\mathrm{viol}}$ and $\Delta^\ell \leq \bar{\Delta}$, Lemma~\ref{lem:feasibility} guarantees:
\begin{equation}
\bar{\phi}(z^{\ell+1}) \leq \bar{\phi}(z^\ell) - \frac{\bar{\mu}\tau_{\mathrm{viol}}}{4}.
\end{equation}
Since $\bar{\phi}$ is bounded below by $\phi_{\min}$, the number of such iterations is at most $N_{\mathrm{viol}}$ given by \eqref{eq:N_viol_cor}.

The total is $K_{\mathrm{feas}} \leq K^* + K_\Delta + N_{\mathrm{viol}}$.
\end{proof}

\end{document}